\documentclass[12pt,a4paper]{article}
\usepackage[utf8]{inputenc}
\usepackage[english]{babel}
\usepackage{amsmath}
\usepackage{amsthm}
\usepackage{amsfonts}
\usepackage{amssymb}
\usepackage{geometry}
\usepackage[pdftex]{graphicx}
\usepackage[longnamesfirst,authoryear]{natbib}
\usepackage{afterpage}

\newtheorem{Theorem}{Theorem}
\newtheorem{Lemma}{Lemma}

\newtheorem{Corollary}{Corollary}

\newtheorem*{Assumption_I1}{Assumption I1}
\newtheorem*{Assumption_I2}{Assumption I2}
\newtheorem*{Assumption_I3}{Assumption I3}
\newtheorem*{Assumption_I3'}{Assumption I3'}
\newtheorem*{Assumption_I4}{Assumption I4}
\newtheorem*{Assumption_D1}{Assumption D1}
\newtheorem*{Assumption_D2}{Assumption D2}
\newtheorem*{Assumption_D3}{Assumption D3}
\newtheorem*{Assumption_D4}{Assumption D4}
\newtheorem*{Assumption_P1}{Assumption P1}
\newtheorem*{Assumption_P2}{Assumption P2}
\newtheorem*{Assumption_P3}{Assumption P3}
\newtheorem*{Assumption_P4}{Assumption P4}
\newtheorem*{Assumption_C1}{Assumption C1}
\newtheorem*{Assumption_C2}{Assumption C2}
\newtheorem*{Assumption_C3}{Assumption C3}
\newtheorem*{Assumption_C4}{Assumption C4}
\newtheorem*{Assumption_C5}{Assumption C5}
\newtheorem*{Assumption_N1}{Assumption N1}
\newtheorem*{Assumption_N2}{Assumption N2}
\newtheorem*{Assumption_N3}{Assumption N3}
\newtheorem*{Assumption_N3'}{Assumption N3'}
\newtheorem{Example}{Example}
\newtheorem{Simulation}{Simulation}

\begin{document}

\author{Takuya Ishihara\footnote{University of Tokyo, Graduate School of Economics, email: ishihara-takuya143@g.ecc.u-tokyo.ac.jp.
I would like to express my appreciation to the editor, the associate editor, and an anonymous referee for their careful reading and comments.
I also would like to thank Katsumi Shimotsu, Hidehiko Ichimura, and the seminar participants at the University of Tokyo, Hakodate Conference, Otaru University of Commerce, and Nanzan University.
This research is supported by Grant-in-Aid for JSPS Research Fellow (17J03043) from the JSPS.}}
\title{Identification and Estimation of Time-Varying Nonseparable Panel Data Models without Stayers}
\date{\today}
\maketitle

\begin{abstract}
This paper explores the identification and estimation of nonseparable panel data models.
We show that the structural function is nonparametrically identified when it is strictly increasing in a scalar unobservable variable, the conditional distributions of unobservable variables do not change over time, and the joint support of explanatory variables satisfies some weak assumptions.
To identify the target parameters, existing studies assume that the structural function does not change over time, and that there are ``stayers", namely individuals with the same regressor values in two time periods.
Our approach, by contrast, allows the structural function to depend on the time period in an arbitrary manner and does not require the existence of stayers.
In estimation part of the paper, we consider parametric models and develop an estimator that implements our identification results. 
We then show the consistency and asymptotic normality of our estimator.
Monte Carlo studies indicate that our estimator performs well in finite samples.
Finally, we extend our identification results to models with discrete outcomes, and show that the structural function is partially identified.

\end{abstract}

\section{Introduction}

This paper considers the identification and estimation of the following nonseparable panel data model:
\begin{eqnarray}
Y_{it} &=& g_t(X_{it},U_{it}),  \ \ \ \ i = 1, \cdots, n , \ \  t=1,\cdots ,T, \label{model}
\end{eqnarray}
where $Y_{it}\in \mathbb{R}$ is a scalar response variable, $X_{it} \in \mathbb{R}^k$ is a vector of explanatory variables, and $U_{it} \in \mathbb{R}$ is a scalar unobservable variable.
Suppose that $\mathbf{Y}_i = (Y_{i1}, \cdots , Y_{iT})$ and $\mathbf{X}_i = (X_{i1}', \cdots , X_{iT}')$ are observable.
Many widely used panel data models fall into this category.
For example, this specification contains the textbook linear panel data model
$$
Y_{it} = X_{it}'\beta + \alpha_i +\epsilon_{it},
$$
because we can regard $\alpha_i + \epsilon_{it}$ as $U_{it}$.
Furthermore, it contains the following nonlinear panel data models:
\begin{eqnarray}
Y_{it} &=& h^{-1}(X_{it}'\beta + \gamma_t + \alpha_i +\epsilon_{it}), \label{abrevaya} \\
Y_{it} &=& c(U(\alpha_i,\epsilon_{it})) + X_{it}' \beta(U(\alpha_i,\epsilon_{it})), \label{galvao}
\end{eqnarray}
where (\ref{abrevaya}) is the transformation model proposed by \cite{abrevaya1999leapfrog}, and (\ref{galvao}) is the random effects quantile regression model proposed by \cite{galvao2017quantile}.

The importance of unobserved heterogeneity when modeling economic behavior is widely recognized.
Nonseparable models capture the unobserved heterogeneity effect of explanatory variables on outcomes because these models allow the derivative of the structural function to depend on an unobserved variable.
Indeed, there is extensive literature on nonseparable panel data models including \cite{altonji2005cross},\cite{evdokimov2010identification}, \cite{hoderlein2012nonparametric}, \cite{chernozhukov2013average}, \cite{d2013nonlinear}, and \cite{chernozhukov2015nonparametric}.

This study shows that we can nonparametrically identify $g_t(x,u)$ when $g_t(x,u)$ is strictly increasing in $u$, the conditional distributions of $U_{it}$ are the same over time, and the support of $\mathbf{X}_i$ satisfies some weak assumptions.
To identify the target parameters, many nonseparable panel data models assume that the structural function does not change over time, and require the existence of ``stayers", namely individuals with the same regressor values in two time periods.
By contrast, our approach allows $g_t$ to depend on the time period $t$ in an arbitrary manner and does not require the existence of stayers.

Although modeling time trends is important in research on panel data, existing nonseparable panel data models assume that the structural function does not change over time or impose some restrictions on these time trends.
For instance, \cite{altonji2005cross} do not allow $g_t$ to depend on time period $t$; \cite{evdokimov2010identification} and \cite{hoderlein2012nonparametric} use additive time effects; and \cite{chernozhukov2013average} and \cite{chernozhukov2015nonparametric} use additive location time effects
and multiplicative-scale time effects. Moreover,
\cite{chernozhukov2015nonparametric} assume that $g_t(x,u)$ can be written as $\mu_t(x) + \sigma_t(x) \phi(x,u)$.
Thus, time effects are linearly conditional on explanatory variables in this model, and as such it does not allow for nonlinear time effects.
Indeed, \cite{d2013nonlinear} allow for nonlinear time effects by assuming that $g_t(x,u)$ can be written as $m_t(h(x,u))$, where $m_t$ is a monotonic transformation.
While this transformation extends the typical additive location time effects and captures macro-shocks, it does not allow the effect of macro-shocks to depend directly on an unobserved variable, and stipulates that $\nabla_x g_t(x,u) / \nabla_u g_t(x,u)$ does not depend on time. 
For example, consider the relationship between consumption and income.
We write the Engel function of the $i$-th household as 
$$
Y_{it} = \phi(X_{it},M_t,U_{it}),
$$
where $Y_{it}$ is consumption, $X_{it}$ is income, $U_{it}$ is the scalar unobserved heterogeneity that represents preference, and $M_t$ is a macroeconomic variable.
However, such a model does not satisfy \cite{d2013nonlinear} since $\nabla_x \phi(x,M_t,u) / \nabla_u \phi(x,M_t,u)$ depends on $M_t$.
By contrast, our assumptions can accommodate this model, because we can rewrite this as (\ref{model}) by treating $\phi(x,M_t,u)$ as $g_t(x,u)$.

Many nonseparable panel data models require the existence of stayers:
\cite{evdokimov2010identification}, \cite{hoderlein2012nonparametric}, and \cite{chernozhukov2015nonparametric} require stayers in order to identify the structural functions or derivatives of the average and quantile structural functions.
In particular, to identify the structural function, \cite{evdokimov2010identification} requires the support of $(X_{i1},X_{i2})$ contains $(x,x)$ for all $x$.
Many empirically important models do not satisfy this assumption.
For example, in standard difference-in-differences (DID) models, there are no individuals treated during both time periods.
Our approach does not require the existence of stayers and allows the support condition employed in standard DID models.

Our identification approach is based on the conditional stationary condition, that is, the conditional distribution function of $U_{it}$ given $\mathbf{X}_i$ does not change over time.
Similar assumptions are employed in all the aforementioned papers except for \cite{altonji2005cross}.
Indeed, \cite{manski1987semiparametric}, \cite{abrevaya1999leapfrog}, \cite{athey2006identification}, \cite{hoderlein2012nonparametric}, \cite{graham2012identification}, \cite{chernozhukov2013average}, and \cite{chernozhukov2015nonparametric} essentially make the same assumption.
Whereas \cite{evdokimov2010identification} does not impose this assumption explicitly, a similar assumption is made by considering the following model: $Y_{it} = m(X_{it},\alpha_i)+U_{it}$.
In this model, the unobservable variable $\alpha_i$ automatically satisfies the conditional stationarity because $\alpha_i$ does not depend on $t$.
By contrast, \cite{d2013nonlinear} do not assume the conditional stationarity of $U_{it}$ given $\mathbf{X}_i$ because they consider the identification of nonseparable models using repeated cross-sections.
Rather, they assume that the conditional distribution of $U_{it}$ given $V_{it} \equiv F_{X_t}(X_{it})$ does not depend on time.

In the literature on nonseparable panel data models, many papers allow the unobservable variable to be a vector or do not impose monotonicity on the structural function, for example, \cite{altonji2005cross}, \cite{evdokimov2010identification}, \cite{hoderlein2012nonparametric}, \cite{chernozhukov2013average}, and \cite{chernozhukov2015nonparametric}.
On the other hand, our model assumes that the unobservable variable is scalar, and that the structural function is strictly increasing in the unobservable variable.
These assumptions are restrictive but crucial for our identification results.

In the estimation part of the paper, we assume that the admissible collection of structural functions is indexed by a finite-dimensional parameter.
In what follows, we develop an estimator based on the conditional stationary condition.
Our method is similar to that of \cite{torgovitsky2017minimum}.
The estimator is obtained by minimizing the distance between the conditional distributions of $U_{i1}$ and $U_{i2}$.
We then prove the consistency and asymptotic normality of this estimator.
Because the asymptotic variance is complicated, we also show the validity of the nonparametric bootstrap.
Monte Carlo studies indicate that our estimator performs well in finite samples.

Finally, we extend our identification results to models in which outcomes are discrete.
This class of models includes many empirically important models such as binary choice panel data models.
Models in this class cannot point-identify $g_t$, but can partially identify it by using the suggestion developed in \cite{chesher2010instrumental}.
We also allows $g_t$ to depend on the time period $t$ in an arbitrary manner and does not require the existence of stayers.
However, the support condition becomes stronger than it is in models with continuous outcomes.

The remainder of the paper is organized as follows.
Section 2 demonstrates the nonparametric identification of $g_t$ when outcome variables are continuous.
In Section 3, we propose the estimator under the parametric assumption and discuss its consistency, asymptotic normality, and bootstrap.
Section 4 reports the results of several Monte Carlo simulations.
In Section 5, we consider the case where $Y_{it}$ is discrete and show the partial identification of $g_t$.
Section 6 offers concluding remarks.
The proofs of the theorems and auxiliary lemmas are collected in the Appendix.

\section{Identification}

First, for notational convenience we drop the subscript $i$ and let $T=2$.
It is straightforward to extend the results to the case with $T \geq 3$.
For any random variables $V$ and $W$, let $F_{V|W}$ and $Q_{V|W}$ denote the conditional distribution function and the conditional quantile function, respectively.
$\mathcal{X}_{t}$, $\mathcal{X}_{12}$, and $\mathcal{U}_t$ denote, respectively, the supports of $X_t$, $(X_1,X_2)$, and $U_t$.

First, we assume $g_t(x,u)$ is strictly increasing in $u$, and $U_t$ is continuously distributed.

\begin{Assumption_I1}
(i) For all $t$, the function $g_t(x,u)$ is continuous and strictly increasing in $u$ for all $x$.
If $X_t$ is continuously distributed, then $g_t(x,u)$ is also continuous in $x$.
(ii) For all $t$, $U_t|\mathbf{X}=\mathbf{x}$ is continuously distributed for all $\mathbf{x}$.
\end{Assumption_I1}

\begin{Assumption_I2}
For all $t$ and $\mathbf{x} \in \mathcal{X}_{12}$, the conditional distribution of the $Y_t$ conditional on $\mathbf{X}=\mathbf{x}$ is continuous and strictly increasing.
\end{Assumption_I2}

Many nonseparable panel data models do not employ the strict monotonicity assumption, for example, \cite{altonji2005cross}, \cite{hoderlein2012nonparametric}, \cite{chernozhukov2013average}, and \cite{chernozhukov2015nonparametric}.
These models allow for unobserved variables to be multivariate.
Hence, our model is more restrictive than theirs.
However, as noted in the previous section, our model covers many widely used panel data models, such as typical linear fixed-effects models.

Assumptions I1 and I2 rule out the case where outcomes are discrete.
In Section 5, we relax the strict monotonicity assumption by allowing $g_t(x,u)$ to be flat inside the support of $U_t$, and consider the case where outcomes are discrete.

Next, we impose the normalization assumption.

\begin{Assumption_I3}
For some $\bar{x} \in \mathcal{X}_1$, we have $g_1(\bar{x},u)=u$ for all $u$.
\end{Assumption_I3}

Assumption I3 is a normalization assumption common in nonseparable models (see, e.g., \cite{matzkin2003nonparametric}).
Because we assume $U_1 | \mathbf{X}=\mathbf{x} \overset{d}{=} U_2 | \mathbf{X}=\mathbf{x}$ below, it is sufficient to normalize $g_1(x,u)$ exclusively.
The functions $g_t(x,u)$ and distributions of $U_t$ depend on the choice of $\bar{x}$.
However, it is easy to show that the function
$$
h_t(x,\tau) \equiv g_t(x,Q_{U_t}(\tau))
$$
does not depend on the choice of $\bar{x}$.

Nevertheless, we can normalize this model in an alternative way.

\begin{Assumption_I3'}
For all $t$, the marginal distribution of $U_t$ is uniform on $[0,1]$.
\end{Assumption_I3'}

Under this normalization and additional assumptions, we can regard this structural function as the quantile function of the potential outcome considered by \cite{chernozhukov2005iv}.
They refer to $U_t$ as the rank variable.
As they show, under the rank invariance or rank similarity assumption, we can think of the function $g_t(x,u)$ as the quantile function of the potential outcome.
It is easy to show that the function $g_t(x,\tau)$ under Assumption I3' is the same as the function $h_t(x,\tau)$ under Assumption I3.

Hereafter, we use Assumption I3, but we can replace Assumption I3 with Assumption I3' and identify the structural function $g_t$, as we show below.

We assume the conditional stationarity of $U_t$ by following \cite{manski1987semiparametric}, \cite{abrevaya1999leapfrog}, \cite{athey2006identification},  \cite{hoderlein2012nonparametric}, \cite{graham2012identification}, \cite{chernozhukov2013average}, and \cite{chernozhukov2015nonparametric}.

\begin{Assumption_I4}
(i) The conditional distributions of the unobservable $U_t$ conditional on $\mathbf{X}$ is the same across $t$.
That is, for all $\mathbf{x} \in \mathcal{X}_{12}$, we have
\begin{equation}
U_1 | \mathbf{X}=\mathbf{x} \overset{d}{=} U_2 | \mathbf{X}=\mathbf{x}, \label{conditional stationary assumption} 
\end{equation}
which implies that $\mathcal{U}_1 = \mathcal{U}_2 \equiv \mathcal{U}$.
(ii) For all $t$, the conditional support of $U_t|\mathbf{X}=\mathbf{x}$ is $\mathcal{U}$.
\end{Assumption_I4}

When we can decompose $U_t$ into time-variant and time-invariant parts, this assumption does not impose any restrictions on the dependence between the time-invariant part and $\mathbf{X}$.
To see this, let $U_t = U(\alpha, \epsilon_t)$, where $\alpha$ is time-invariant and $\epsilon_t$ is time-variant.
Then, Assumption I4 holds, if 
\begin{equation}
\epsilon_1|\alpha=a, \mathbf{X}=\mathbf{x} \, \overset{d}{=} \, \epsilon_2|\alpha=a, \mathbf{X}=\mathbf{x}. \label{conditional stationarity assumption 2}
\end{equation}
Because condition (\ref{conditional stationarity assumption 2}) allows $\alpha$ to be correlated with $\mathbf{X}$ arbitrarily, Assumption I4 imposes no restrictions on the time-invariant unobservable variables.

Indeed, \cite{evdokimov2010identification} and \cite{d2013nonlinear} employ similar assumptions, although the former does not make this assumption explicitly.
By considering the model $Y_{it} = m(X_{it},\alpha_i)+U_{it}$, the unobservable variable $\alpha_i$ automatically satisfies the conditional stationarity.
Moreover, since \cite{d2013nonlinear} consider the identification using repeated cross-sections, they do not impose this assumption.
Instead, they impose the following:
$$
U_1|\mathbf{V}_1=\mathbf{v} \overset{d}{=} U_2|\mathbf{V}_2=\mathbf{v},
$$
where $\mathbf{V}_t \equiv (F_{X_{t,1}}(X_{t,1}), \cdots , F_{X_{t,k}}(X_{t,k}))$ and $\mathbf{v} \in (0,1)^k$.

To show the identification of $g_t$, we introduce the following sets.
Define $\mathcal{S}_0^1 \equiv \{\bar{x}\}$, $\mathcal{S}_0^2 \equiv \{x\in\mathcal{X}_2: (\bar{x},x) \in \mathcal{X}_{12}\}$, namely the cross-section of $\mathcal{X}_{12}$ at $X_1=\overline{x}$.
For $m \geq 1$, define
\begin{eqnarray}
\mathcal{S}_m^1 &\equiv & \{ x\in \mathcal{X}_1 : \text{there exists $x_2 \in \mathcal{S}_{m-1}^2$ such that $(x,x_2)\in\mathcal{X}_{12}$.} \}, \nonumber \\
\mathcal{S}_m^2 &\equiv & \{ x\in \mathcal{X}_2 : \text{there exists $x_1 \in \mathcal{S}_{m}^1$ such that $(x_1,x)\in\mathcal{X}_{12}$.} \}. \nonumber
\end{eqnarray}
Figure 1 illustrates these sets.
Because $U_1|\mathbf{X}=\mathbf{x} \overset{d}{=} U_2|\mathbf{X}=\mathbf{x}$, for all $(x_1,x_2) \in \mathcal{X}_{12}$, we have
\begin{eqnarray}
F_{Y_1|\mathbf{X}}\left( g_1(x_1,u)|x_1,x_2 \right) &=& P\left( g_1(x_1,U_1) \leq g_1(x_1,u)|X_1=x_1,X_2=x_2 \right) \nonumber \\
&=& P\left( U_1 \leq u |X_1=x_1,X_2=x_2 \right) \nonumber \\
&=& P\left( U_2 \leq u |X_1=x_1,X_2=x_2 \right) \nonumber \\
&=& P\left( g_2(x_2,U_2) \leq g_2(x_2,u) |X_1=x_1,X_2=x_2 \right) \nonumber \\
&=& F_{Y_2|\mathbf{X}}\left( g_2(x_2,u)|x_1,x_2 \right). \nonumber
\end{eqnarray}
Because $F_{Y_t|\mathbf{X}}(y|x_1,x_2)$ is invertible in $y$ for all $(x_1,x_2) \in \mathcal{X}_{12}$, we obtain 
\begin{eqnarray}
g_1(x_1,u) &=& Q_{Y_1|\mathbf{X}}\left( F_{Y_2|\mathbf{X}}\left( g_2(x_2,u)|x_1,x_2 \right) |x_1,x_2 \right) \nonumber \\
g_2(x_2,u) &=& Q_{Y_2|\mathbf{X}}\left( F_{Y_1|\mathbf{X}}\left( g_1(x_1,u)|x_1,x_2 \right) |x_1,x_2 \right). \label{QF}
\end{eqnarray}

Equations (\ref{QF}) imply that if $g_1(x_1,u)$ (or $g_2(x_2,u)$) is identified and $(x_1,x_2) \in \mathcal{X}_{12}$, then $g_2(x_2,u)$ (or $g_1(x_1,u)$) is also identified.
First, we can identify $g_1(\bar{x},u)$ because $g_1(\bar{x},u)=u$ holds by Assumption I3.
Hence, we can identify $g_2(x,u)$ for all $x \in \mathcal{S}_0^2$, because $g_2(x,u) = Q_{Y_2|\mathbf{X}}\left( F_{Y_1|\mathbf{X}}\left( g_1(\bar{x},u)|\bar{x},x \right) |\bar{x},x \right) = Q_{Y_2|\mathbf{X}}\left( F_{Y_1|\mathbf{X}}\left( u|\bar{x},x \right) |\bar{x},x \right)$.
We now turn to identifying $g_1(x,u)$ for $x \in \mathcal{S}_1^1$.

First, we fix $x \in \mathcal{S}_1^1$.
According to the definition of $\mathcal{S}_1^1$, there exists $x_2 \in \mathcal{S}_0^2$ such that $(x,x_2) \in \mathcal{X}_{12}$.
Then, it follows from (\ref{QF}) that 
$$
g_1(x,u) = Q_{Y_1|\mathbf{X}}\left( F_{Y_2|\mathbf{X}}\left( g_2(x_2,u)|x,x_2 \right) |x,x_2 \right),
$$
and hence, $g_1(x,u)$ is identified because $g_2(x_2,u)$ is already identified.
Similarly, by using (\ref{QF}), we can identify $g_2(x,u)$ for all $x \in \mathcal{S}_1^2$.
Repeating this argument provides the following theorem.

\begin{Theorem}
Suppose that Assumptions I1, I2, I3, and I4 are satisfied.
For all $t$, if we have $\mathcal{X}_t = \overline{\cup_{m=0}^{\infty} \mathcal{S}_m^t}$, then the structural function $g_t(x,u)$ is identified for all $x\in\mathcal{X}_t$ and $u\in \mathcal{U}$.
\end{Theorem}

We also show the identification of $g_t$ under Assumption I3' instead of I3.

\begin{Corollary}
Suppose that Assumptions I1, I2, I3', and I4 are satisfied.
For all $t$, if $\mathcal{X}_t = \overline{\cup_{m=0}^{\infty} \mathcal{S}_m^t}$ holds for some $\overline{x} \in \mathcal{X}_1$, then the function $g_t(x,u)$ is identified for all $x\in\mathcal{X}_t$ and $u\in \mathcal{U}$.
\end{Corollary}

This identification approach is similar to that of \cite{d2015identification}, \cite{torgovitsky2015identification}, and \cite{ishihara2017partial}, who all identify nonseparable IV models.
\cite{d2015identification} and \cite{ishihara2017partial} use the same normalization as Assumption I3'.
\cite{d2015identification} show that under appropriate assumptions, if for all $x$ and $x'$, we identify the function $T_{x',x}(y)$ that is strictly increasing in $y$ and satisfies
$$
g(x',u) = T_{x',x}(g(x,u)),
$$
then we can identify the structural function $g(x,u)$.
We can also construct similar functions and show that $g_t$ is point identified.

We next introduce some examples that satisfy this support condition.

\begin{Example}[DID model]
In standard DID models, if $X_t$ is a treatment indicator, then we have $\mathcal{X}_{12}=\{(0,0),(0,1)\}$.
Because $\mathcal{X}_1=\{0\}$, we assume $\bar{x}=0$.
That is, $g_1(0,u)=u$ for all $u$.
Hence, we identify $g_1(x,u)$ for all $x \in \mathcal{X}_1$ and $u \in \mathcal{U}$.
Then, because $\mathcal{S}^2_0=\{0,1\}=\mathcal{X}_{12}$, the support condition of Theorem 1 holds and we can identify $g_2(x,u)$ for all $x \in \mathcal{X}_2$ and $u \in \mathcal{U}$.

Our identification approach does not require the joint distribution of $(Y_1,Y_2)$.
Hence, if we can observe $D\equiv \mathbf{1}\{X_2=1\}$, then we can identify the structural function $g_t$ by using repeated cross-sections.
If the potential outcome $Y_t(x)$ is equal to $g_t(x,U_t)$, then this setting is similar to \cite{athey2006identification}.

Similar to \cite{athey2006identification}, we can also identify the counterfactual distribution even when $\mathcal{X}_{12} \neq \{(0,0),(0,1)\}$.
Let $Y_t(x) \equiv g_t(x,U_t)$ denote the potential outcomes.
Then, we can identify $F_{Y_2(x)|X_2}(y|x')$, where $x \neq x'$.
Suppose that there exists $x_1 \in \mathcal{X}_1$ such that $(x_1,x), (x_1,x') \in \mathcal{X}_{12}$.
In this case, it follows from (\ref{QF}) that
\begin{eqnarray}
& & F_{Y_1|\mathbf{X}=(x_1,x')}\left( Q_{Y_1|\mathbf{X}=(x_1,x)}(F_{Y_2|\mathbf{X}=(x_1,x)}(y)) \right) \nonumber \\
&=& F_{Y_1|\mathbf{X}}\left( g_1(x_1,g_2^{-1}(x,y)) |x_1,x' \right) \nonumber \\
&=& P\left( g_1(x_1,U_1) \leq g_1(x_1,g_2^{-1}(x,y)) | X_1=x_1, X_2=x' \right) \nonumber \\
&=& P\left( g_2(x,U_1) \leq y | X_1=x_1, X_2=x' \right) \nonumber \\
&=& P\left( Y_2(x) \leq y | X_1=x_1, X_2=x' \right). \nonumber 
\end{eqnarray}
Hence, we can obtain $F_{Y_2(x)|X_2}(y|x')$ by integrating out $x_1$.
The left-hand side is similar to the counterfactual distribution of \cite{athey2006identification}.
When $\mathcal{X}_{12} = \{(0,0), (0,1)\}$, this result is same as their result.
\end{Example}

\begin{Example}[connected support]
When the interior of $\mathcal{X}_{12}$ is connected, the support condition of Theorem 1 holds.
Because the interior of $\mathcal{X}_{12}$ is connected, for any $x \in \mathcal{X}_1$, there exists a series $(x_1^0,x_2^0),(x_1^1,x_2^0),(x_1^1,x_2^1),(x_1^2, x_2^1),(x_1^2,x_2^2) \cdots$ such that $x_1^0=\bar{x}$, $(x_1^m,x_2^m), (x_1^{m+1},x_2^m) \in \mathcal{X}_{12}$ for all $m = 0, 1, \cdots$, and $\lim_{m \rightarrow \infty} x_1^m = x$.
Figure 2 illustrates this result intuitively.
From the definition of $\mathcal{S}^1_m$, $x_1^m \in \mathcal{S}_1^m$ for all $m$.
Hence, we have $x \in \overline{\cup_{m=0}^{\infty} \mathcal{S}_m^1}$, and the support condition of Theorem 1 holds.
\end{Example}

The support condition of Theorem 1 rules out the case where $X_1=X_2$.
Hence, if the explanatory variables do not vary across time periods, such as sex or race, this support condition does not hold.

If we have panel data with more than two periods, we can relax this support condition.
Similar to the case where $T=2$, we define the following sets.
Define $\tilde{\mathcal{S}}^1_0 \equiv \{\bar{x}\}$, $\tilde{\mathcal{S}}^t_0 \equiv \{x \in \mathcal{X}_t : (\bar{x},x) \in supp(X_1,X_t)\}$, $t=2, \cdots , T$, and for $m \geq 1$ and $t = 1, \cdots , T$,
\begin{eqnarray}
\tilde{\mathcal{S}}^t_m \equiv \bigcup_{s \neq t} \{ x \in \mathcal{X}_t : \text{there exists $x_s \in \tilde{\mathcal{S}}^s_{m-1}$ such that $(x_s,x) \in supp(X_s,X_t)$.} \}. \nonumber
\end{eqnarray}
Then, we show that $g_t(x,u)$ is point-identified under a similar support condition to that of Theorem 1.

\begin{Corollary}
Suppose Assumptions I1, I2, I3, and I4 are satisfied for $T\geq 3$.
For $t = 1, \cdots, T$, if $\mathcal{X}_t = \overline{\cup_{m=0}^{\infty} \tilde{\mathcal{S}}^t_m}$, then the function $g_t(x,u)$ is identified for all $x \in \mathcal{X}_t$ and $u \in \mathcal{U}$.
\end{Corollary}

\section{Estimation and Inference}

In the previous sections, we considered nonparametric identification.
In this section, we assume that the admissible collection of structural functions is indexed by a finite-dimensional parameter.
Consider the following parametric model:
\begin{equation}
Y_{it} = g_t(X_{it},U_{it};\theta_0) \ \ \ \ \ i = 1, \cdots, n, \ \  t=1, \cdots, T. \label{model_p}
\end{equation}
The outcome functions are parameterized by $\theta \in \Theta \subset \mathbb{R}^{d_{\theta}}$, where $\theta_0 \in \Theta$ is the true parameter.
We assume that $\{(\mathbf{Y}_i, \mathbf{X}_i)\}_{i=1}^n$ are independent and identically distributed.

Indeed, \cite{torgovitsky2017minimum} consider a similar setting, and develop an estimator based on the identification result of \cite{torgovitsky2015identification}.
Following \cite{torgovitsky2017minimum}, we develop a minimum distance estimator based on our identification results.

The following assumptions are the parametric versions of Assumptions I1, I2, I3, and I4.

\begin{Assumption_P1}
(i) For all $t$, the function $g_t(x,u;\theta)$ is continuous and strictly increasing in $u$ for all $\theta \in \Theta$.
(ii) For all $t$, $U_{it}|\mathbf{X}_i=\mathbf{x}$ is continuously distributed for all $\mathbf{x}$.
\end{Assumption_P1}

\begin{Assumption_P2}
For all $t$ and $\mathbf{x} \in supp(\mathbf{X})$, the conditional distribution of $Y_{it}$ conditional on $\mathbf{X}_i=\mathbf{x}$ is continuous and strictly increasing.
\end{Assumption_P2}

\begin{Assumption_P3}
(i) For some $\bar{x} \in \mathcal{X}_1$, $g_1(\bar{x},u;\theta)=u$ holds for all $u\in \mathcal{U}$ and $\theta \in \Theta$.
(ii) For all $\theta, \theta' \in \Theta$ with $\theta \neq \theta'$, we have $g_t(\cdot,\cdot;\theta) \neq g_t(\cdot,\cdot;\theta')$ for some $t$.
\end{Assumption_P3}

\begin{Assumption_P4}
(i) For all $\mathbf{x} \in supp(\mathbf{X})$ and $s,t \in \{1, \cdots, T\}$, we have $U_{is} | \mathbf{X}_i=\mathbf{x} \overset{d}{=} U_{it} | \mathbf{X}_i=\mathbf{x}$.
(ii) The support of $U_{it}|\mathbf{X}_i=\mathbf{x}$ is $\mathcal{U}$.
\end{Assumption_P4}

These assumptions are similar to the assumptions in Section 2.
Assumption P3(ii) allows that $g_t(x,u;\theta)$ does not depend on some part of $\theta$.
For example, consider $\theta = (\theta_1,\theta_2, \cdots, \theta_T)$.
Then, this condition allows that $g_t$ depends exclusively on $\theta_t$.

Similar to Section 2, we suppose $T=2$.
Under Assumptions P1--P4 and the support condition of Theorem 1, it follows from Theorem 1 that
\begin{equation}
U_{i1,\theta} | \mathbf{X}_i = \mathbf{x} \overset{d}{=}  U_{i2,\theta} | \mathbf{X}_i = \mathbf{x} \ \ \text{for all } \mathbf{x} \ \ \Leftrightarrow \ \ \theta = \theta_0, \label{identification condition}
\end{equation}
where $U_{it,\theta} \equiv g_t^{-1}(X_{it},Y_{it};\theta)$, $t = 1,2$.
Therefore, (\ref{identification condition}) implies that the function
\begin{eqnarray}
D_{\theta}(v) &\equiv & P\left( U_{i1,\theta} \leq v_u, \ \mathbf{X}_i \leq v_{\mathbf{x}} \right) - P\left( U_{i2,\theta} \leq v_u, \ \mathbf{X}_i \leq v_{\mathbf{x}} \right) \nonumber \\
&=& E\left[ \left( \mathbf{1}\{Y_{i1} \leq g_1(X_{i1},v_u;\theta)\}-\mathbf{1}\{Y_{i2} \leq g_2(X_{i2},v_u;\theta)\} \right) \mathbf{1}\{\mathbf{X}_i \leq v_{\mathbf{x}}\} \right] \label{D}
\end{eqnarray}
is zero for all $v=(v_{\mathbf{x}},v_u) \in \mathcal{V} \equiv \mathcal{X}_{12} \times \mathcal{U}$ if and only if $\theta=\theta_0$.
Let $\|\cdot\|_{\mu}$ denote the $L_2$-norm with respect to a probability measure with support $\mathcal{V}$.
Then, we have $\|D_{\theta}\|_{\mu} \geq 0$ and
\begin{equation}
\|D_{\theta}\|_{\mu} = 0 \ \ \Leftrightarrow \ \ \theta=\theta_0. \label{identification}
\end{equation}
Hence, $\theta_0$ is the value that provides a global minimum for $\|D_{\theta}\|_{\mu}$.

We construct the estimator $\hat{D}_{n,\theta}(v)$ of $D_{\theta}(v)$ as a sample analogue of (\ref{D}):
\begin{equation}
\hat{D}_{n,\theta}(v) \equiv \frac{1}{n} \sum_{i=1}^n \left( \mathbf{1}\{Y_{i1} \leq g_1(X_{i1},v_u;\theta)\}-\mathbf{1}\{Y_{i2} \leq g_2(X_{i2},v_u;\theta)\} \right) \mathbf{1}\{\mathbf{X}_i \leq v_{\mathbf{x}}\}. \label{D_hat}
\end{equation}
This is a natural estimator of $D_{\theta}(v)$.
We can obtain the estimator $\hat{\theta}_n$ by minimizing $\|\hat{D}_{n,\theta}\|_{\mu}$.
That is,
\begin{equation}
\hat{\theta}_n = \text{arg}\min_{\theta \in \Theta} \|\hat{D}_{n,\theta}\|_{\mu}. \label{estimator}
\end{equation}
This estimator is similar to the estimators proposed by \cite{brown2002weighted} and \cite{torgovitsky2017minimum}.
They prove the consistency and the asymptotic normality of their estimators, and also show the consistency of the nonparametric bootstrap.
In what follow, we likewise prove the consistency and the asymptotic normality of our estimator, and show that the nonparametric bootstrap is consistent.

First, we collect the observable data together into a single vector, $W_i=(\mathbf{Y}_i,\mathbf{X}_i)=(Y_{i1},Y_{i2},X_{i1},X_{i2})$.
Next, we define 
$$
A_{\theta}^v(w) \equiv \left[ \mathbf{1}\{y_1 \leq  g_1(x_1,v_u;\theta)\}-\mathbf{1}\{y_2 \leq g_2(x_2,v_u;\theta)\} \right] \mathbf{1}\{\mathbf{x} \leq v_{\mathbf{x}}\},
$$
where $w = (y_1,y_2,x_1,x_2)$.
Then, $\hat{D}_{n,\theta}(v) = \frac{1}{n} \sum_{i=1}^n A_{\theta}^v(W_i)$.

\subsection{Consistency}

Here, we demonstrate the consistency of $\hat{\theta}_n$.
Under condition (\ref{identification}), the following assumptions are sufficient for $\hat{\theta}_n$ to be consistent.

\begin{Assumption_C1}
$\hat{\theta}_n$ satisfies $\|\hat{D}_{n,\hat{\theta}_n}\|_{\mu} = \inf_{\theta \in \Theta} \|\hat{D}_{n,\theta}\|_{\mu}$.
\end{Assumption_C1}

\begin{Assumption_C2}
$\Theta$ is compact.
\end{Assumption_C2}

\begin{Assumption_C3}
For all $\theta',\theta$, $|g_t(x,u;\theta')-g_t(x,u;\theta)|\leq \bar{g}(x)\|\theta'-\theta\|$ holds for some strictly positive $\bar{g}(x)$ with $E[\bar{g}(X_t)] \leq K$, where $0<K<\infty$.
\end{Assumption_C3}

\begin{Assumption_C4}
For all $t$, $Y_{it}$ is absolutely continuously distributed given $\mathbf{X}_i$, with a conditional pdf $f_{Y_t|\mathbf{X}}(y|\mathbf{x})$ that is uniformly bounded above and continuous in $y$.
\end{Assumption_C4}

\begin{Assumption_C5}
For all $t$, there exists an integer $J_t$ and functions $\{\beta_j\}_j^{J_t}$ such that for every $\theta \in \Theta$ and $u \in \mathcal{U}$ there is an $\alpha^t(\theta,u) \in \mathbb{R}^{J_t}$ with $g_t(x,u)= \sum_{j=1}^{J_t} \alpha_j^t(\theta, u) \beta_j(x)$.
\end{Assumption_C5}

Assumption C1 entails that $\hat{\theta}_n$ minimizes $\|\hat{D}_{n,\theta}\|_{\mu}$.
Assumptions C3 and C4 imply that $\|D_{\theta}\|_{\mu}$ is continuous in $\theta$.
Assumption C5 ensures that a class of functions, $\{A_{\theta}^v:\theta \in \Theta, v \in \mathcal{V}\}$, is $P$-Glivenko--Cantelli.
Hence, we show that $\|\hat{D}_{n,\theta}\|_{\mu}$ uniformly convergences to $\|D_{\theta}\|_{\mu}$ almost surely.
\cite{brown2002weighted} and \cite{torgovitsky2017minimum} also make similar assumptions.
From these results and the compactness of $\Theta$, we show the consistency from the usual arguments of extremum estimators (e.g., \cite{newey1994large}).

\begin{Theorem}
Under Assumptions P1--P4, C1--C5, and (\ref{identification}), we have $\hat{\theta}_n \rightarrow_{a.s.} \theta_0$.
\end{Theorem}

\subsection{Asymptotic Normality}

Because the objective function $\|\hat{D}_{n,\theta}\|_{\mu}$ is not differentiable in $\theta$, our approach follows from \cite{pakes1989simulation}.
Similarly, although $\hat{D}_{n,\theta}(v)$ is not differentiable in $\theta$, we also assume $D_{\theta}(v)$ is differentiable in $\theta$.
We let $\nabla_x f$ denote the column vector of partial derivatives of $f$ with respect to $x$.
We define $\Gamma_{\theta}(v) \equiv \nabla_\theta D_{\theta}(v)$ and $\Gamma_{0}(v) \equiv \Gamma_{\theta_0}(v)$.

\begin{Assumption_N1}
$\theta_0$ is an interior point of $\Theta$.
\end{Assumption_N1}

\begin{Assumption_N2}
For all $t$, $g_t(x,u;\theta)$ is continuously differentiable in $\theta$ in the neighborhood of $\theta_0$.
In the neighborhood of $\theta_0$, $| \nabla_{\theta} g_t(x,u;\theta) |$ is bounded by some positive function $\nabla\bar{g}(x)$ with $E|\nabla \bar{g}(X_t)| < \infty$.
\end{Assumption_N2}

\begin{Assumption_N3}
(i) There exists $c >0$ such that $\| \Gamma_0(v)'a \|_{\mu} \geq c \|a\|$ for all $a \in \mathbb{R}^{d_{\theta}}$.
(ii) $\{\Gamma_{\theta}(v):v\in \mathcal{V}\}$ is equicontinuous in $\theta$ at $\theta_0$.
(iii) $\int \|\Gamma_0(v)\|^2 d\mu(v) < \infty$.
\end{Assumption_N3}

Assumption N1 is a standard assumption.
Combined with Assumption C4, Assumption N2 implies that $D_{\theta}(v)$ is continuously differentiable in $\theta$ in the neighborhood of $\theta_0$.
Assumption N3(i) is a rank condition that
corresponds to Assumption D4 in \cite{torgovitsky2017minimum}.
Assumption N3(ii) implies that $\Gamma_0(v)'(\theta - \theta_0)$ approximates $D_{\theta}(v)$ in the neighborhood of $\theta_0$ uniformly over $v$.

\begin{Theorem}
Under assumptions P1--P4, C1--C5, N1--N3, and (\ref{identification}),
$$
\sqrt{n}(\hat{\theta}_n - \theta_0) \rightsquigarrow N(0,\Delta_0^{-1}\Sigma_0\Delta_0^{-1}),
$$
where $\Delta_0 \equiv \int \Gamma_0(v)\Gamma_0(v)'d\mu(v)$ and
$$
\Sigma_0 \equiv \int \int_{\mathcal{V}\times \mathcal{V}} \left\{ \Psi(v,v') \Gamma_0(v) \Gamma_0(v')' \right\} d\mu(v) d\mu(v')
$$
with $\Psi(v,v') \equiv E[A_{\theta_0}^v(W)A_{\theta_0}^{v'}(W)]$.
\end{Theorem}

The proof of Theorem 3 is similar to the proof in \cite{pakes1989simulation} for their Theorem (3.3).

The asymptotic distribution of $\sqrt{n}(\hat{\theta}_n-\theta_0)$ depends on the probability measure $\mu$.
\cite{carrasco2000generalization} consider the generalized method of moments (GMM) procedure with a continuum of moment conditions, obtaining the optimal estimator.
They consider the following type of GMM estimator to minimize
$$
\int \int \hat{D}_{n,\theta}(v) a_n(v,v') \hat{D}_{n,\theta}(v') dv dv',
$$
where $a_n(v,v')$ converges to a kernel $a(v,v')$.
As in \cite{torgovitsky2017minimum}, we consider only the special case where $a_n(v,v')=a(v,v')$ and $a(v,v')=0$ for $v\neq v'$.
Although our setting appears to be similar to that of \cite{carrasco2000generalization}, their approach is not directly applicable because their objective function is smooth.
Hence, we do not pursue this problem.

\subsection{Bootstrap}

Let $\{W_{in}^*\}_{i=1}^n$ denote a bootstrap sample drawn with replacement from $\{W_i\}_{i=1}^n$.
That is, $\{W_{in}^*\}_{i=1}^n$ are independently and identically distributed from the empirical measure $P_n$, conditional on the realizations $\{W_i\}_{i=1}^n$.
We define
$$
\hat{D}^*_{n,\theta}(v) \equiv \frac{1}{n} \sum_{i=1}^n A_{\theta}^v(W_{in}^*)
$$
as the bootstrap counterpart to $\hat{D}_{n,\theta}(v)$.
Next, we suppose that $\hat{\theta}_n^*$ satisfies
\begin{equation}
\|\hat{D}^*_{n,\hat{\theta}_n^*} \|_{\mu} = \inf_{\theta \in \Theta} \|\hat{D}^*_{n,\theta}\|. \label{bootstrap estimator}
\end{equation}
Then, we can obtain the following theorem.

\begin{Theorem}
Suppose that $\hat{\theta}_n^*$ satisfies (\ref{bootstrap estimator}).
Under the assumptions of Theorem 4, $\sqrt{n}(\hat{\theta}_n^*-\hat{\theta}_n)$ converges weakly to the limit distribution of $\sqrt{n}(\hat{\theta}_n-\theta_0)$ in probability.
\end{Theorem}

The proof of this theorem is similar to the proof of Theorem 6 in \cite{brown2002weighted}.
From Theorem 3, we show that
\begin{eqnarray}
\hat{\theta}_n - \theta_0 &=& \gamma_n + o_p(n^{-1/2}), \nonumber 
\end{eqnarray}
where $\gamma_n = \Delta_0^{-1} \cfrac{1}{n} \sum_{i=1}^n \int \Gamma_0(v) (A_{\theta_0}^v(W_i)-E[A_{\theta_0}^v(W_i)]) d\mu(v)$.
By using the bootstrap stochastic equicontinuity due to \cite{gine1990bootstrapping}, we show that
$$
\sqrt{n} \| \hat{\theta}_n^*-\hat{\theta}_n - \gamma^*_n \|
$$
converges to zero in probability, conditional on almost all samples, where $\gamma^*_n$ is the bootstrap counterpart of $\gamma_n$.
The term $\gamma_n^*$ has the same limiting distribution as $\gamma_n$ according to the bootstrap theorem for the mean in $\mathbb{R}^{d_{\theta}}$.
Hence, we show that $\sqrt{n}(\hat{\theta}_n^*-\hat{\theta}_n)$ converges weakly to the limit distribution of $\sqrt{n}(\hat{\theta}_n-\theta_0)$ in probability.

\section{Simulations}

To evaluate the finite sample performance of our estimator, we conducted two Monte Carlo experiments.

\begin{Simulation}
The outcome equation is given by
\begin{eqnarray}
g_1(x,u) &=& u + (\theta_1 + \theta_2 u) (x-\bar{x}), \nonumber \\
g_2(x,u) &=& u + (\theta_1 + \theta_3 u) (x-\bar{x}), \nonumber
\end{eqnarray}
where $\bar{x}=0$.
Because $g_1(\bar{x},u) = u$ for all $x$, the normalization Assumption I3 is satisfied.
We assume
\begin{eqnarray}
X_t &=& 4\Phi(Z_t) \ \ \ t=1,2 , \nonumber \\
U_t &=& \alpha + \epsilon_t \ \ \ t=1,2, \nonumber
\end{eqnarray}
where $\Phi(\cdot)$ is the standard normal distribution function and
\begin{eqnarray}
(Z_1,Z_2,\alpha) &\sim & N\left( \left( \begin{array}{c}
0 \\ 0 \\ 0
\end{array} \right) ,\left( \begin{array}{ccc}
1.0 & 0.3 & 0.6 \\ 0.3 & 1.0 & 0.5 \\ 0.6 & 0.5 & 1.0
\end{array} \right) \right), \nonumber \\
(\epsilon_1,\epsilon_2) &\sim & N\left( \left( \begin{array}{c}
0 \\ 0
\end{array} \right) ,\left( \begin{array}{cc}
(0.3)^2 & 0 \\ 0 & (0.3)^2
\end{array} \right) \right). \nonumber
\end{eqnarray}
Because the correlations between $\alpha$ and $(Z_1$, $Z_2)$ are not zero, $X_1$ and $X_2$ are correlated with $U_t$.
Because $\epsilon_1|\mathbf{X}=\mathbf{x}, \alpha=a \overset{d}{=} \epsilon_2|\mathbf{X}=\mathbf{x}, \alpha=a$, the conditional stationarity assumption holds.
We used $\mu = \text{Unif}(0,4) \times \text{Unif}(0,4) \times N(0,1)$ as the integrating measure.

We considered the following two settings: (i) $(\theta_1,\theta_2,\theta_3) = (0.5,1.0,0.7)$, (ii) $\theta_2 = \theta_3$ and $(\theta_1,\theta_2) = (0.5,1.0)$.
Under Setting (i), we cannot use estimation methods of other papers, because their time effects depend on $U_t$.
On the other hand, under Setting (ii), there are no time effects.
Hence, we can estimate $E\left[ \nabla_x g_t(X_t,U_t)|X_1=X_2=x \right]$ by using the method proposed in \cite{hoderlein2012nonparametric} because there are stayers.
Thus, we estimated $E\left[ \nabla_x g_t(X_t,U_t)|X_1=X_2=2 \right]$ using our method and their method and compared the results of both.
Under Setting (ii), we have $E\left[ \nabla_x g_t(X_t,U_t)|X_1=X_2=2 \right] = \theta_1 +\theta_2 E[U_t|X_1=X_2=2] = 0.5$.
Hence, we estimated $E\left[ \nabla_x g_t(X_t,U_t)|X_1=X_2=2 \right]$ by using $\hat{\theta}_1 + \hat{\theta}_2 \hat{E}[U_t|X_1=X_2=2]$.

Table 1 contains the results under Setting (i) for sample sizes of $400$, $800$, and $1600$.
The number of replications was set to $1000$ throughout.
Table 1 shows the bias, standard deviation, and the mean squared error (MSE) of the estimates of $(\theta_1,\theta_2, \theta_3)$, highlighting that the standard deviation and MSE decrease as the sample size increases.
In some cases, the biases of the estimates do not decrease.
However they are relatively small under all settings.

Table 2 contains the results under Setting (ii) for sample sizes of $500$ and $1000$.
Table 2 shows that the standard error of our estimator is smaller than that of \cite{hoderlein2012nonparametric} for all settings.
Although the bias of our estimator is larger than their estimator, the MSE of our estimator is smaller. 
\end{Simulation}

\begin{Simulation}[DID model]
We considered the case where $\mathcal{X}_{12}=\{(0,0),(0,1)\}$.
The outcome equation is given by
\begin{eqnarray}
g_1(x,u) &=& u  \nonumber \\
g_2(x,u) &=& (\theta_1 + \theta_2 u)(1-x) +(\theta_3 + \theta_4 u)x, \nonumber
\end{eqnarray}
where $(\theta_1,\theta_2,\theta_3,\theta_4) = (0.5,0.7,0.5,1.2)$.
Because $g_1(x,u)$ does not depend on $x$, the normalization Assumption I3 holds for any $\bar{x} \in \mathcal{X}_1$.
We assumed
\begin{eqnarray}
X_2 &=& \mathbf{1}\{Z>0\}, \nonumber \\
U_t &=& \alpha + \epsilon_t \ \ \ t=1,2, \nonumber
\end{eqnarray}
where $\Phi(\cdot)$ is the standard normal distribution function and
\begin{eqnarray}
(Z,\alpha) &\sim & N\left( \left( \begin{array}{c}
0 \\ 2.0
\end{array} \right) ,\left( \begin{array}{cc}
1.0 & 0.6  \\ 0.6 & 1.0 
\end{array} \right) \right), \nonumber \\
(\epsilon_1,\epsilon_2) &\sim & N\left( \left( \begin{array}{c}
0 \\ 0
\end{array} \right) ,\left( \begin{array}{cc}
(0.3)^2 & 0 \\ 0 & (0.3)^2
\end{array} \right) \right). \nonumber
\end{eqnarray}
Because $\epsilon_1|X_2=x \overset{d}{=} \epsilon_2|X_2=x$ for all $\mathbf{x}$, the conditional stationarity Assumption I4 holds.
When $X_2=0$, we have $Y_2=g_2(0,U_2)=\theta_1 + \theta_2 U_2$, and when $X_2=1$, we have $Y_2=g_2(1,U_2)=\theta_3 + \theta_4 U_2$.
This specification is similar to that of a typical DID model.
However, letting $Y_t(x)=g_t(x,U_2)$ be potential outcomes, this model does not satisfy the parallel trend assumption if $\theta_2 \neq 1$, because $E[Y_2(0)-Y_1(0)|X_2=x] = \theta_1 + (\theta_2-1)E[U_1|X_2=x]$ holds by the conditional stationarity of $U_t$.
Hence, we cannot estimate the average treatment effect on the treated (ATT) or the average treatment effect (ATE) by using a standard DID method.
Under this setting, we have
\begin{eqnarray}
ATE & \equiv & E[Y_2(1)-Y_2(0)] = 1.00, \nonumber \\
QTE_{25} & \equiv & Q_{Y_2(1)}(0.25)-Q_{Y_2(0)}(0.25)\fallingdotseq 0.65, \nonumber \\
QTE_{50} & \equiv & Q_{Y_2(1)}(0.50)-Q_{Y_2(0)}(0.50)=1.00, \nonumber \\
QTE_{75} & \equiv & Q_{Y_2(1)}(0.75)-Q_{Y_2(0)}(0.75)\fallingdotseq 1.35. \nonumber 
\end{eqnarray}
We also estimated ATE and QTE as follows:
\begin{eqnarray}
\hat{ATE} &=& \left( \hat{\theta}_3 + \hat{\theta}_4 \hat{E}[\hat{U}] \right) - \left( \hat{\theta}_1 + \hat{\theta}_2 \hat{E}[\hat{U}] \right), \nonumber \\
\hat{QTE}_{100\tau} &=& \left( \hat{\theta}_3 + \hat{\theta}_4 \hat{Q}_{\hat{U}}(\tau) \right) - \left( \hat{\theta}_1 + \hat{\theta}_2 \hat{Q}_{\hat{U}}(\tau) \right), \nonumber
\end{eqnarray}
where $\hat{E}[\hat{U}]$ is a sample average of $\hat{U} \equiv (Y_{i1}, \cdots, Y_{in}, g_2^{-1}(X_{12},Y_{12};\hat{\theta}), \cdots, g_2^{-1}(X_{n2},Y_{n2};\hat{\theta}))$, and $\hat{Q}_{\hat{U}}(\tau)$ is a sample $\tau$-th quantile of $\hat{U}$.
Because $X_1=0$ and $X_2$ is discrete, we used
$$
A_{\theta}^v(w) = \left( \mathbf{1}\{y_1 \leq v_u \} - \mathbf{1}\{ y_2 \leq g_2(x_2,v_u; \theta ) \} \right) \mathbf{1}\{x_2 = v_x\},
$$
where $v = (v_x,v_u)$.
We used $\mu = Ber(0.5) \times N(\bar{Y}_1,s_{Y_1})$ as the integrating measure, where $\bar{Y}_1$ is the sample average of $Y_1$ and $s_{Y_1}$ is the standard deviation of $Y_1$.
Table 4 contains the results of this experiment for sample sizes of $400$, $800$, and $1600$.
The number of replications was set to $1000$ throughout.
Table 4 shows the bias, standard deviation, and MSE of the estimates of parameters, the ATE, and QTE, highlighting that the standard deviation and MSE of estimates decrease as the sample size increases.
The biases of the estimates of parameters, the ATE, and $QTE_{50}$ are relatively small, whereas the biases of the estimates of $QTE_{25}$ and $QTE_{75}$ are large.
This may be caused by the fact that the sample quantiles are biased.
\end{Simulation}

\section{Discrete Outcomes}

In this section, we consider the case where outcomes are discrete.
In the case of discrete outcomes, we cannot point-identify $g_t(x,u)$.
This is likewise true in \cite{athey2006identification}, \cite{chesher2010instrumental}, and \cite{ishihara2017partial}.
They consider the case where outcomes are discrete, and instead show partial identification of the structural function.
Hence, in this section, we also consider partial identification of $g_t(x,u)$.

First, we drop the $i$ subscript and let $T=2$, as in Section 2.
Let $\mathcal{Y}_t$ denote the support of $Y_t$.
The assumptions employed in Section 2 do not allow the outcomes to be discrete.
Hence, we impose the following assumptions.

\begin{Assumption_D1}
(i) For all $t\in \{1,2\}$, the function $g_t(x,u)$ is weakly increasing in $u$ for all $x$.
(ii) For all $t\in \{1,2\}$, $U_t|\mathbf{X}=\mathbf{x}$ is continuously distributed for all $\mathbf{x}$.
\end{Assumption_D1}

\begin{Assumption_D2}
(i) For all $t \in \{1,2\}$, $Y_t$ is discretely distributed.
(ii) $\mathcal{Y}_1 = \mathcal{Y}_2 \equiv \mathcal{Y}$ with $\underline{y} \equiv \inf \mathcal{Y}$ and $\overline{y} \equiv \sup \mathcal{Y}$.
\end{Assumption_D2}

\begin{Assumption_D3}
For all $t \in \{1,2\}$, the marginal distribution of $U_t$ is uniform on $[0,1]$.
\end{Assumption_D3}

\begin{Assumption_D4}
(i) For all $\mathbf{x} \in \mathcal{X}_{12}$, $U_1 | \mathbf{X}=\mathbf{x} \overset{d}{=} U_2 | \mathbf{X}=\mathbf{x}$ holds.
(ii) The support of $U_t|\mathbf{X}=\mathbf{x}$ is $[0,1]$.
\end{Assumption_D4}

Assumption I1(i) stipulates that $g_t(x,u)$ is strictly increasing in $u$.
If $U_t$ is continuously distributed, then $Y_t$ must be continuously distributed under Assumption I1(i).
Hence, in this section, we relax Assumption I1 by allowing $g_t$ to be flat inside the support of $U_t$.
\cite{athey2006identification} and \cite{chesher2010instrumental} also employ this weakly monotonic assumption in models with discrete outcomes.
Furthermore, when outcomes are discrete, we cannot use Assumption I3, because $U_t$ is continuously distributed.
Hence, we use another normalization assumption.
Assumption D4 is identical to Assumption I4.

We can thus obtain the following theorem.

\begin{Theorem}
Suppose that Assumptions D1, D2, D3, and D4 are satisfied.
For all $t \in \{1,2\}$, if $\mathcal{X}_{12}=\mathcal{X}_1 \times \mathcal{X}_2$ holds, then we have
\begin{eqnarray}
g_t(x,u) &\geq & g^L_t(x,u) \equiv \inf \{y \in [\underline{y},\overline{y}] :G_{t,x}^L(y) \geq u\}, \nonumber \\
g_t(x,u) &\leq & g^U_t(x,u) \equiv \sup \{y\in [\underline{y},\overline{y}]: G^U_{t,x}(y) \leq u\}, \nonumber
\end{eqnarray}
where $G^L_{t,x}$ and $G^U_{t,x}$ are defined by (\ref{def GL}) and (\ref{def GU}), respectively.
\end{Theorem}

This identification approach is similar to that in \cite{ishihara2017partial}, who considers the identification of nonseparable models with binary instruments and shows that the structural functions are partially identified when outcomes are discrete.

In Theorem 5, we assume that $\mathcal{X}_{12}=\mathcal{X}_1 \times \mathcal{X}_2$.
Although this support condition does not require stayers, it is nevertheless stronger than that of Theorem 1.
Indeed, we can relax this condition and partially identify $g_t$ under a weaker support condition.
However, if we do, then the bounds of $g_t$ may be looser.

To illustrate Theorem 5, we introduce two examples.

\begin{Example}[DID model with binary outcomes]
Suppose that the outcomes are binary and $\mathcal{X}_{12} = \{(0,0),(0,1)\}$.
Then, $\mathcal{X}_{12} = \mathcal{X}_1 \times \mathcal{X}_2$, where $\mathcal{X}_1=\{0\}$ and $\mathcal{X}_2 = \{0,1\}$.
This is the usual DID setting.
Define $D\equiv \mathbf{1}\{X_2=1\}$.
We consider the partial identification of $g_2(0,u)$ and $g_2(1,u)$.

In this case, we have 
$$
G^L_{2,0}(y) = P(D=1) F^+_{Y_2|D=1}(T^U_{2,1,0}(y)) + P(D=0) F^+_{Y_2|D=0}(T^U_{2,0,0}(y)),
$$
where $T^U_{2,1,0}(y) = Q^+_{Y_2|D=1} \circ F^+_{Y_2|D=1} \circ Q^+_{Y_1|D=0} \circ F^+_{Y_1|D=0}(y)$ and $T^U_{2,0,0}(y)=y$.
We define $p_t(d) \equiv P(Y_t=1|D=d)$, then
$$
G^L_{2,0}(y)= \begin{cases}
P(D=1,Y_2 \leq \mathbf{1}\{p_2(1) \geq p_1(1)\}) & \text{if $y<0$} \\
P(D=1,Y_2 \leq \mathbf{1}\{p_2(1) \geq p_1(1) \ \text{or} \ p_1(0) \geq p_2(0) \}) + P(D=0,Y_2=0) & \text{if $0\leq y<1$} \\
1 & \text{if $y \geq 1$}
\end{cases}.
$$
Therefore, we can obtain a lower bound
$$
g^L_2(0,u) = \begin{cases}
\mathbf{1}\{u > P(Y_2=0)\} & \text{if $p_1(1) > p_2(1)$ and $p_1(0) < p_2(0)$} \\
\mathbf{1}\{u > P(Y_2=0) + P(D=1,Y_2=1)\} & \text{if $p_1(1) \leq p_2(1)$ or $p_1(0) \geq p_2(0)$}
\end{cases}.
$$
Similarly, we can obtain the following functions:
\begin{eqnarray}
g^L_2(1,u) &=& \begin{cases}
\mathbf{1}\{u > P(Y_2=0)\} & \text{if $p_1(1) < p_2(1)$ and $p_1(0) > p_2(0)$} \\
\mathbf{1}\{u > P(D=0) + P(D=1,Y_2=0)\} & \text{if $p_1(1) \geq p_2(1)$ or $p_1(0) \leq p_2(0)$}
\end{cases} \nonumber \\
g^U_2(0,u) &=& \begin{cases}
\mathbf{1}\{u \geq P(D=0,Y_2=0)\} & \text{if $p_1(1) \geq p_2(1)$ or $p_1(0) \leq p_2(0)$} \\
\mathbf{1}\{u \geq P(Y_2=0)\} & \text{if $p_1(1) < p_2(1)$ and $p_1(0) > p_2(0)$}
\end{cases} \nonumber \\
g^U_2(1,u) &=& \begin{cases}
\mathbf{1}\{u \geq P(D=1,Y_2=0)\} & \text{if $p_1(1) \leq p_2(1)$ or $p_1(0) \geq p_2(0)$} \\
\mathbf{1}\{u \geq P(Y_2=0)\} & \text{if $p_1(1) > p_2(1)$ and $p_1(0) < p_2(0)$} 
\end{cases} \nonumber
\end{eqnarray}

If we define the potential outcomes as $Y_t(x) = g_t(x,U_t)$, we can partially identify the ATE.
Because $g_2^L(x,u)$ and $g_2^U(x,u)$ respectively denote the lower and upper bounds of $g_2(x,u)$, we have
$$
E[g_2^L(x,U)] \leq E[Y_2(x)] \leq E[g_2^U(x,U)], \ \ \text{for $x=1,2$},
$$
where $U \sim U(0,1)$.
Hence, we have
$$
E[g_2^L(1,U)]-E[g_2^U(0,U)] \leq \mu_{ATE} \leq E[g_2^U(1,U)] - E[g_2^L(0,U)],
$$
where $\mu_{ATE} \equiv E[Y_2(1)-Y_2(0)]$.

Hence, above bounds of $g_2$ imply that the lower (upper) bound of ATE is not larger (smaller) than 0.
Actually, when $p_1(1) < p_2(1)$ and $p_1(0) > p_2(0)$, that is $E[Y_1(0)|D=1] < E[Y_2(1)|D=1]$ and $E[Y_1(0)|D=0] > E[Y_2(0)|D=0]$, a lower bound of ATE becomes 0.
This situation implies that the mean of the treated group increases, although the time trend effect is negative.
Hence, in this case, it is intuitive that the ATE is larger than 0.
Contrarily, when $p_1(1) > p_2(1)$ and $p_1(0) < p_2(0)$, that is $E[Y_1(0)|D=1] > E[Y_2(1)|D=1]$ and $E[Y_1(0)|D=0] < E[Y_2(0)|D=0]$, an upper bound of ATE becomes 0.
This situation implies that the mean of the treated group decreases, although the time trend effect is positive.
Hence, in this case, it is intuitive that the ATE is smaller than 0.

As an example, we consider the following case:
\begin{eqnarray}
&E[Y_1|D=1] = 0.4,  &E[Y_1|D=0] = 0.3, \nonumber \\
&E[Y_2|D=1] = 0.5,  &E[Y_1|D=0] = 0.2, \nonumber \\
&P(D=1) = 0.5. & \nonumber
\end{eqnarray}
In this case, we can obtain 
\begin{eqnarray}
g_2^L(0,u) &=& \mathbf{1}\{u>0.9\}, \nonumber \\
g_2^L(1,u) &=& \mathbf{1}\{u>0.65\}, \nonumber \\
g_2^U(0,u) &=& \mathbf{1}\{u>0.65\}, \nonumber \\
g_2^U(1,u) &=& \mathbf{1}\{u>0.25\}. \nonumber 
\end{eqnarray}
Hence, in this case, ATE is smaller than 0.65 and larger than 0.
As discussed above, because $E[Y_1|D=1] < E[Y_2|D=1]$ and $E[Y_1|D=0] > E[Y_2|D=0]$, a lower bound of ATE becomes 0.
\end{Example}

\begin{Example}
We consider the following model:
$$
Y_t = g_t(X_t,U_t) = \mathbf{1}\{U_t > (1+\exp(\alpha_t + \beta_t X_t))^{-1}\}, \ \ t = 1,2,
$$
where $U_t = \Phi(\epsilon_t)$ and
$$
(X_1,X_2,\epsilon_t) \sim  N\left( \left( \begin{array}{c}
0 \\ 0 \\ 0
\end{array} \right) ,\left( \begin{array}{ccc}
1.0 & 0.6 & 0.4 \\ 0.6 & 1.0 & 0.4 \\ 0.4 & 0.4 & 1.0
\end{array} \right) \right).
$$
Hence, $U_t \sim U(0,1)$ for all $t$ and $U_1|\mathbf{X}=\mathbf{x} \overset{d}{=} U_2|\mathbf{X}=\mathbf{x}$ for all $\mathbf{x}$.
We set $(\alpha_1,\alpha_2,\beta_1,\beta_2)=(0,0.3,0.5,0.6)$.
Under this setting, we calculate $g_t^L(x,u)$ and $g_t^U(x,u)$ defined by Theorem 5 for $x = -2, -1, 0, 1, 2$.
Table 4 shows $g_t^L(x,u)$, $g_t^U(x,u)$, and $g_t(x,u)$ at $x=-2,-1,0,1,2$.

When $x$ is small, the lower (upper) bounds are uninformative (informative).
Contrarily, when $x$ is large, lower (upper) bounds are informative (uninformative).
In this model, there is a positive time trend because $g_1(x,u) \leq g_2(x,u)$.
These bounds reflect this fact.
That is, they also satisfy $g_1^L(x,u) \leq g_2^L(x,u)$ and $g_1^U(x,u) \leq g_2^U(x,u)$.
\end{Example}

We can extend Theorem 5 to panel data with more than two periods.

\begin{Corollary}
Suppose Assumptions D1, D2, D3, and D4 are satisfied for $T \geq 3$.
For $t=1, \cdots, T$, if $supp(\mathbf{X}) = \mathcal{X}_1 \times \cdots \times \mathcal{X}_T$, then we have
\begin{eqnarray}
g_t(x,u) & \geq & g_t^L(x,u) \equiv \inf \{y \in [\underline{y},\overline{y}] : G_{t,x}^L(y) \geq u\}, \nonumber \\
g_t(x,u) & \leq & g_t^U(x,u) \equiv \sup \{y \in [\underline{y},\overline{y}] : G_{t,x}^U(y) \leq u\}, \nonumber
\end{eqnarray}
where $G_{t,x}^L(y)$ and $G_{t,x}^U(y)$ are defined by (\ref{GL and GU T>=3}).
\end{Corollary}

\section{Conclusion}

This paper explored the identification and estimation of nonseparable panel data models.
We showed that the structural function is nonparametrically identified when the structural function $g_t(x,u)$ is strictly increasing in $u$, the conditional distributions of $U_{it}$ are the same over time, and the joint support of $\mathbf{X}_i$ satisfies weak assumptions.
Many nonseparable panel data models assume that the structural function does not change over time and that stayers exist.
By contrast, our approach allows the structural function to depend on the time period in an arbitrary manner, and it does not require the existence of stayers.
In estimation part of the paper, we assumed that the admissible collection of structural functions is indexed by a finite-dimensional parameter.
We developed an estimator that implements our identification results. 
We demonstrated the consistency and asymptotic normality of our estimator and showed the validity of the nonparametric bootstrap.
Monte Carlo studies indicated that our estimator performs well with finite samples.
Finally, we extended our identification results to models with discrete outcomes and showed that the structural function is partially identified.

\afterpage{\clearpage}
\newpage
\section*{Appendix 1: The case with $T \geq 3$}

Here, we consider the estimation when $T \geq 3$.
Then, similar to the case of $T=2$, we show that under the assumptions of Corollary 2,
\begin{equation}
U_{is,\theta}|\mathbf{X}_i=\mathbf{x} \overset{d}{=} U_{it,\theta}|\mathbf{X}_i=\mathbf{x} \ \ \text{for all $\mathbf{x}$ and $s,t$} \ \ \Leftrightarrow \ \ \theta = \theta_0, \label{identification condition T>2}
\end{equation}
where $U_{it,\theta} \equiv g_t^{-1}(X_{it},Y_{it};\theta)$.
(\ref{identification condition T>2}) implies that
\begin{eqnarray}
E \left[ \mathbf{1}\{U_{it,\theta} \leq v_u, \ \mathbf{X}_i \leq v_{\mathbf{x}} \} \right] - E \left[ \mathbf{1}\{U_{is,\theta} \leq v_u, \ \mathbf{X}_i \leq v_{\mathbf{x}} \} \right] \nonumber
\end{eqnarray}
is zero for all $v \in \mathcal{V}$ and $t,s \in \{1, \dots, T\}$ if and only if $\theta = \theta_0$.
Therefore, we have
$$
D^t_{\theta}(v) = 0, \ \text{for all $v$ and $t$} \ \ \Leftrightarrow \ \ \theta = \theta_0,
$$
where
\begin{eqnarray}
D^t_{\theta}(v) & \equiv & E \left[ \mathbf{1}\{U_{t,\theta} \leq v_u, \ \mathbf{X} \leq v_{\mathbf{x}} \} \right]-\frac{1}{T} \sum_{s=1}^T E \left[ \mathbf{1}\{U_{s,\theta} \leq v_u, \ \mathbf{X} \leq v_{\mathbf{x}} \} \right]. \nonumber
\end{eqnarray}
Then, we show that $\frac{1}{T} \sum_{t=1}^T \| D^t_{\theta} \|_{\mu}^2 \geq 0$ and
\begin{equation}
\frac{1}{T} \sum_{t=1}^T \| D^t_{\theta} \|_{\mu}^2 = 0 \ \ \Leftrightarrow \ \ \theta = \theta_0.  \label{identification T>2}
\end{equation}
Hence, $\theta_0$ is the value that provides a global minimum for $\frac{1}{T} \sum_{t=1}^T \| D^t_{\theta} \|_{\mu}^2$.

We construct the estimator $\hat{D}_{n,\theta}^t(v)$ of $D_{\theta}^t(v)$ as a sample analogue of $D_{\theta}^t(v)$:
\begin{eqnarray}
\hat{D}_{n,\theta}^t(v) &=& \frac{1}{n} \sum_{i=1}^n \mathbf{1}\{ Y_{it} \leq g_t(X_{it},v_u;\theta), \ \mathbf{X}_i \leq v_{\mathbf{x}} \} \nonumber \\
& & - \frac{1}{T} \sum_{s=1}^T \left( \frac{1}{n} \sum_{i=1}^n \mathbf{1}\{ Y_{is} \leq g_s(X_{is},v_u;\theta), \ \mathbf{X}_i \leq v_{\mathbf{x}} \} \right) \nonumber \\
&=& \frac{1}{n} \sum_{i=1}^n \left( \mathbf{1}\{ Y_{it} \leq g_t(X_{it},v_u;\theta)\} - \frac{1}{T} \sum_{s=1}^T \mathbf{1}\{ Y_{is} \leq g_s(X_{is},v_u;\theta)\} \right) \mathbf{1}\{ \mathbf{X}_i \leq v_{\mathbf{x}} \}. \nonumber 
\end{eqnarray}
We obtain the estimator $\hat{\theta}_n$ of $\theta_0$ by minimizing $\frac{1}{T} \sum_{t=1}^T \|\hat{D}_{n,\theta}^t\|_{\mu}^2$.
When $T=2$, this estimator is identical to estimator (\ref{estimator}).
Define
$$
A_{\theta}^{t,v}(W_i) \equiv \left( \mathbf{1}\{ Y_{it} \leq g_t(X_{it},v_u;\theta)\} - \frac{1}{T} \sum_{s=1}^T \mathbf{1}\{ Y_{is} \leq g_s(X_{is},v_u;\theta)\} \right) \mathbf{1}\{ \mathbf{X}_i \leq v_{\mathbf{x}} \},
$$
then $\hat{D}^t_{n,\theta}(v) = \frac{1}{n} \sum_{i=1}^n A_{\theta}^{t,v}(W_i)$.

Then, similar to Theorems 3 and 4, the following theorems hold.

\begin{Theorem}
Suppose that Assumptions P1--P4, C2--C5, and (\ref{identification T>2}) hold.
If 
\begin{equation}
\frac{1}{T} \sum_{t=1}^T \| \hat{D}^t_{n,\hat{\theta}_n} \|_{\mu}^2 = \inf_{\theta \in \Theta} \frac{1}{T} \sum_{t=1}^T \| \hat{D}^t_{n,\theta} \|_{\mu}^2 \label{Def estimator T>2}
\end{equation}
holds, then $\hat{\theta}_n \rightarrow_{a.s.} \theta_0$.
\end{Theorem}

We define $\Gamma^t_{\theta}(v) \equiv \nabla_{\theta} D^t_{\theta}(v)$ and $\Gamma^t_0(v) \equiv \Gamma^t_{\theta_0}(v)$.

\begin{Assumption_N3'}
(i) There exists $c >0$ such that $\sum_{t=1}^T \| \Gamma_0^t(v)'a \|_{\mu}^2 \geq c \|a\|^2$ for all $a \in \mathbb{R}^{d_{\theta}}$.
(ii) For all $t \in \{1, \cdots, T\}$, $\{\Gamma_{\theta}^t(v):v\in \mathcal{V}\}$ is equicontinuous in $\theta$ at $\theta_0$.
\end{Assumption_N3'}

\begin{Theorem}
Suppose that Assumptions P1--P4, C2--C5, N1, N2, N3', (\ref{identification T>2}), and (\ref{Def estimator T>2}) hold.
Then, $\sqrt{n}(\hat{\theta}_n - \theta_0) \rightsquigarrow N(0,\Delta_{0,T}^{-1} \Sigma_{0,T} \Delta_{0,T}^{-1})$ holds, where $\Delta_{0,T} \equiv \frac{1}{T} \sum_{t=1}^T \int \Gamma^t_0(v) \Gamma^t_0(v)' d\mu(v)$ and
$$
\Sigma_{0,T} \equiv \frac{1}{T^2} \sum_{s=1}^T \sum_{t=1}^T \int \int_{\mathcal{V} \times \mathcal{V}} \{ \Psi_{s,t}(v,v') \Gamma_0^s(v) \Gamma_0^t(v')' \} d\mu(v) d\mu(v')
$$
with $\Psi_{s,t}(v,v') \equiv E[A^{s,v}_{\theta_0}(W) A^{t,v'}_{\theta_0}(W)]$.
\end{Theorem}

\section*{Appendix 2: Proofs}

\begin{proof}[Proof of Theorem 1]
First, we show that $g_t(x,u)$ is identified for all $x \in \cup_{m=0}^{\infty} \mathcal{S}_m^t$.
By the monotonicity of $g_t$ and (\ref{conditional stationary assumption}), equations (\ref{QF}) hold for all $(x_1,x_2)\in \mathcal{X}_{12}$.
First, we can identify $g_1(\overline{x},u)=u$ by Assumption I3.
We can also identify $g_2(x_2,u)$ for all $x_2 \in \mathcal{S}_0^2$ because $(\bar{x},x_2) \in \mathcal{X}_{12}$ and we have 
$$
g_2(x_2,u)= Q_{Y_2|\mathbf{X}}\left( F_{Y_1|\mathbf{X}}\left( g_1(\bar{x},u)|\bar{x},x_2 \right) |\bar{x},x_2 \right) = Q_{Y_2|\mathbf{X}}\left( F_{Y_1|\mathbf{X}}\left( u|\bar{x},x_2 \right) |\bar{x},x_2 \right).
$$
We now turn to identifying $g_1(x_1,u)$ for $x_1 \in \mathcal{S}_1^1$.
Fix $x_1 \in \mathcal{S}_1^1$.
According to the definition of $\mathcal{S}_1^1$, there exists $x_2 \in \mathcal{S}_0^2$ such that $(x_1,x_2) \in \mathcal{X}_{12}$.
Then, it follows from (\ref{QF}) that 
$$
g_1(x_1,u) = Q_{Y_1|\mathbf{X}}\left( F_{Y_2|\mathbf{X}}\left( g_2(x_2,u)|x_1,x_2 \right) |x_1,x_2 \right),
$$
and hence, $g_1(x_1,u)$ is identified because $g_2(x_2,u)$ is already identified.
Similarly, by using (\ref{QF}), we can identify $g_2(x,u)$ for all $x \in \mathcal{S}_1^2$.
Repeating this argument gives the identification of $g_t(x,u)$ for all $x \in \cup_{m=0}^{\infty} \mathcal{S}_m^t$.

Next, we show that $g_t(x,u)$ is identified for all $x \in \mathcal{X}_t$.
We fix $x' \in \mathcal{X}_t \backslash \left( \cup_{m=0}^{\infty} \mathcal{S}_m^t \right)$.
Since $\mathcal{X}_t = \overline{\cup_{m=0}^{\infty} \mathcal{S}_m^t}$, there exists a sequence $\{x^m\}_{m=1}^{\infty} \subset \cup_{m=0}^{\infty} \mathcal{S}_m^t$ such that $\lim_{m\rightarrow \infty} x^m = x'$.
By the continuity of $g_t$, we have $\lim_{m\rightarrow \infty} g_t(x^m,u) = g_t(x',u)$ for all $u \in \mathcal{U}$.
Hence, we can also identify $g_t(x',u)$ because $g_t(x^m,u)$ is identified for all $m$.
\end{proof}
\vspace{0.1in}

\begin{proof}[Proof of Corollary 1]
First, we show that if for all $x,x' \in \mathcal{X}_t$, we can identify the strictly increasing function $T_{t,x',x}(y)$ that satisfies
\begin{equation}
g_t(x',u) = T_{t,x',x}\left( g_t(x,u) \right), \label{T_t,x',x}
\end{equation}
then, $g_t(x,u)$ is point identified.
We define $G^t_x(y) \equiv \int F_{Y_t|X_t}\left( T_{t,x',x}(y) | x'\right) dF_{X_t}(x')$, and then we have
\begin{eqnarray}
G^t_x\left( g_t(x,u) \right) &=& \int F_{Y_t|X_t}\left( g_t(x',u) | x'\right) dF_{X_t}(x') \nonumber \\
&=& \int P\left( U_t \leq u | X_t=x'\right) dF_{X_t}(x') \nonumber \\
&=& P( U_t \leq u ) = u, \nonumber
\end{eqnarray}
where the last equality follows from Assumption I3'.
Because $T_{t,x',x}(y)$ is strictly increasing in $y$, $G^t_x(y)$ is invertible.
Hence, we obtain $g_t(x,u) = \left( G^t_x \right)^{-1}(u)$.
This implies that $g_t(x,u)$ is identified if we can construct $T_{t,x',x}(y)$ for all $x,x' \in \mathcal{X}_t$.

To construct $T_{t,x',x}(y)$, we show that for all $x \in \mathcal{X}_t$, we can identify the strictly increasing function $T^*_{t,x}(y)$ that satisfies 
\begin{equation}
g_t(x,u)=T^*_{t,x}(g_1(\bar{x},u)). \label{T^*}
\end{equation}
For all $x \in \cup_{m=0}^{\infty} \mathcal{S}_m^t$, the proof of Theorem 1 implies that we can construct $T^*_{t,x}(y)$ that satisfies (\ref{T^*}).
Because $F_{Y_t|\mathbf{X}}$ and $Q_{Y_t|\mathbf{X}}$ are strictly increasing, $T^*_{t,x}(y)$ is strictly increasing in $y$ for all $x \in \cup_{m=0}^{\infty} \mathcal{S}_m^t$.
We fix $x' \in \mathcal{X}_t \backslash \left( \cup_{m=0}^{\infty} \mathcal{S}_m^t \right)$.
Since $\mathcal{X}_t = \overline{\cup_{m=0}^{\infty} \mathcal{S}_m^t}$, there exists a sequence $\{x^m\}_{m=1}^{\infty} \subset \cup_{m=0}^{\infty} \mathcal{S}_m^t$ such that $\lim_{m\rightarrow \infty} x^m = x'$.
By the continuity of $g_t$ and (\ref{T^*}), we have $\lim_{m \rightarrow \infty}T^*_{t,x_m}(g_1(\overline{x},u)) = g_t(x,u)$.
Because $g_t(x,u)$ is strictly increasing in $u$, $\lim_{m \rightarrow \infty}T^*_{t,x_m}(y)$ is also strictly increasing in $y$.
Hence, for all $x \in \mathcal{X}_t$, we can identify the strictly increasing function $T^*_{t,x}(y)$ that satisfies (\ref{T^*}).

By using $T^*_{t,x}(y)$, we identify $T_{t,x',x}(y)$ that satisfies (\ref{T_t,x',x}).
Because, for $x,x' \in \mathcal{X}_t$, we have
$$
g_t(x',u) = T^*_{t,x'}\left( (T^{*}_{t,x})^{-1} (g_t(x,u)) \right),
$$
we can construct the function $T_{t,x',x}(y)$ that satisfies $g_t(x',u)=T_{t,x',x}(g_t(x,u))$.
Therefore, we can identify $g_t(x,u)$.
\end{proof}
\vspace{0.1in}

\begin{proof}[Proof of Corollary 2]
The proof is the same as that for Theorem 1.
\end{proof}
\vspace{0.1in}

\begin{proof}[Proof of Theorem 2]
We fix $\delta >0$.
By Lemma 1, C2, and (\ref{identification}), there exists $\epsilon >0$ such that $\|\theta-\theta_0\| \geq \delta$ implies $\|D_{\theta}\|_{\mu} \geq \epsilon$.
Therefore, we have
$$
\|D_{\hat{\theta}_n} \|_{\mu} < \epsilon \ \ \Rightarrow \ \ \|\hat{\theta}_n -\theta_0\| < \delta
$$
and it will suffice to show that $\|D_{\hat{\theta}_n}\|_{\mu} \rightarrow_{a.s.} 0$.
By (\ref{Glivenko-Cantelli}), we have
\begin{equation}
\sup_{\theta} \|\hat{D}_{n,\theta}-D_{\theta}\|_{\mu} \leq  \sup_{\theta,v} | \hat{D}_{n,\theta}(v)-D_{\theta}(v) |   = o_{a.s.}(1). \label{uniformly convergence D}
\end{equation}
By the triangle inequality and C1,
\begin{eqnarray}
\|D_{\hat{\theta}_n}\|_{\mu} &\leq & \| \hat{D}_{n,\hat{\theta}_n} - D_{\hat{\theta}_n}\|_{\mu} + \| \hat{D}_{n,\hat{\theta}_n}\|_{\mu} \nonumber \\
&\leq & \| \hat{D}_{n,\hat{\theta}_n} - D_{\hat{\theta}_n}\|_{\mu} + \| \hat{D}_{n,\theta_0}\|_{\mu}. \nonumber
\end{eqnarray}
By the uniform convergence (\ref{uniformly convergence D}), $\| \hat{D}_{n,\hat{\theta}_n} - D_{\hat{\theta}_n}\|_{\mu} = o_{a.s.}(1)$ and $\| \hat{D}_{n,\theta_0}\|_{\mu} = \| \hat{D}_{n,\theta_0} - D_{\theta_0} \|_{\mu} = o_{a.s.}(1)$.
Hence, we can show that $\hat{\theta}_n \rightarrow_{a.s.} \theta_0$.
\end{proof}
\vspace{0.1in}

\begin{proof}[Proof of Theorem 3]
First, we prove the $\sqrt{n}$-consistency of $\hat{\theta}_n$.
As seen in the previous theorem, $\hat{\theta}_n$ is a consistent estimator of $\theta_0$.
Because $\hat{\theta}_n$ is consistent, we can select a sequence $\{\delta_n\}$ that converges to zero sufficiently slowly to ensure
$$
P(\|\hat{\theta}_n-\theta_0\| \geq \delta_n) \rightarrow 0.
$$
For this sequence, the supremum in (\ref{stochastic equicontinuity D}) runs over a range that includes $\hat{\theta}_n$.
Hence, by the triangle inequality and Lemma 4, we have
$$
\|D_{\hat{\theta}_n}\|_{\mu}-\|\hat{D}_{n,\hat{\theta}_n}\|_{\mu}-\|\hat{D}_{n,\theta_0}\|_{\mu} \leq \|\hat{D}_{n,\hat{\theta}_n} - D_{\hat{\theta}_n} -\hat{D}_{n,\theta_0}\|_{\mu} = o_p(n^{-1/2}).
$$
From Assumption C1,
$$
\|D_{\hat{\theta}_n}\|_{\mu}\leq o_p(n^{-1/2}) + 2\|\hat{D}_{n,\theta_0}\|_{\mu}.
$$
Because $E[A_{\theta_0}^v(W)]=0$ for all $v$ and $E|A_{\theta_0}^v(W) A_{\theta_0}^{v'}(W)| \leq 1$ for all $v$ and $v'$, we have $\sqrt{n}\hat{D}_{n,\theta_0}(v) = \frac{1}{n} \sum_{i=1}^n A_{\theta}^v(W_i) \rightsquigarrow N(0,\Psi(v,v))$.
Since the proof for Lemma 2 shows that $\{A_{\theta_0}^v: v\in\mathcal{V} \}$ is a Donsker class, $\{\sqrt{n}\hat{D}_{n,\theta_0}(v):v\in \mathcal{V}\}$ converges weakly in $l^{\infty}(\mathcal{V})$ to a mean-zero Gaussian process with covariance function $\Psi(v,v')$ and we have $\|\hat{D}_{n,\theta_0}\|_{\mu} = O_p(n^{-1/2})$.
Hence, we have
$$
\|D_{\hat{\theta}_n}\|_{\mu} = O_p(n^{-1/2}).
$$
Because $D_{\theta_0}(v) = 0$ for all $v$, Lemma 3 implies that for all $\theta$ in a neighborhood of $\theta_0$,
\begin{eqnarray}
\|D_{\theta}\|_{\mu} &=& \| \Gamma_0(v)'(\theta-\theta_0) - \left( D_{\theta}(v) - D_{\theta_0}(v) - \Gamma_0(v)'(\theta-\theta_0)  \right) \|_{\mu} \nonumber \\
&\geq & \|\Gamma_0(v)'(\theta-\theta_0)\|_{\mu} - \|D_{\theta}(v) - D_{\theta_0}(v) - \Gamma_0(v)'(\theta-\theta_0) \|_{\mu} \nonumber \\
&\geq & (c-o(1)) \times \|\theta-\theta_0\|. \nonumber
\end{eqnarray}
Therefore, $\|\hat{\theta}_n-\theta_0\| \leq \frac{1}{c-o_p(1)} \|D_{\hat{\theta}_n}\|_{\mu} = O_p(n^{-1/2})$.

Next we establish the asymptotic normality of $\sqrt{n} (\hat{\theta}_n-\theta_0)$ by approximating $\hat{D}_{n,\theta}(v)$ as the linear function
$$
\hat{L}_{n,\theta}(v) \equiv \Gamma_0(v)'(\theta-\theta_0) + \hat{D}_{n,\theta_0}(v).
$$
We have
\begin{eqnarray}
\|\hat{D}_{n,\hat{\theta}_n} - \hat{L}_{n,\hat{\theta}_n} \|_{\mu} & \leq & \|\hat{D}_{n,\hat{\theta}_n} - D_{\hat{\theta}_n} -\hat{D}_{n,\theta_0}\|_{\mu} + \|D_{\hat{\theta}_n}(v)-\Gamma_0(v)'(\hat{\theta}_n-\theta_0)\|_{\mu} \nonumber \\
& \leq & o_p(n^{-1/2}) + o_p(\|\hat{\theta}_n-\theta_0\|) = o_p(n^{-1/2}), \nonumber
\end{eqnarray}
where the second inequality follows from Lemma 3 and Lemma 4, and the last equality follows from the $\sqrt{n}$-consistency of $\hat{\theta}_n$.

Let $\tilde{\theta}_n$ be the value that provides a global minimum for $\|\hat{L}_{n,\theta}\|$.
Then, $\Gamma_0(\cdot)'(\tilde{\theta}_n-\theta_0)$ is the $L_2(\mu)$-projection of $-\hat{D}_{n,\theta_0}(\cdot)$ onto the subspace of $L_2(\mu)$ spanned by the components of $\Gamma_0(\cdot)$.
Because $\Delta_0 = \int \Gamma_0(v)\Gamma_0(v)' d\mu(v)$ is finite and invertible by N3, we have
\begin{equation}
\sqrt{n} (\tilde{\theta}_n-\theta_0) = - \Delta_0^{-1} \int \Gamma_0(v) \sqrt{n} \hat{D}_{n,\theta_0}(v) d\mu(v). \label{theta_tilde}
\end{equation}
Then, we have
\begin{eqnarray}
\int \Gamma_0(v) \sqrt{n} \hat{D}_{n,\theta_0}(v) d\mu(v) &=& \frac{1}{\sqrt{n}} \sum_{i=1}^{n} \int A_{\theta_0}^v(W_i) \Gamma_0(v) d\mu(v) \nonumber \\
&=& \frac{1}{\sqrt{n}} \sum_{i=1}^{n} \xi_i, \nonumber
\end{eqnarray}
for $\xi_i \equiv \int A_{\theta_0}^v(W_i) \Gamma_0(v) d\mu(v)$.
By Fubini's theorem, $E[ \xi_i ] = \int E[A_{\theta_0}^v(W_i)] \Gamma_0(v) d\mu(v) =0$, and
\begin{eqnarray}
E[\xi_i \xi_i'] &=& \int \int_{\mathcal{V}\times \mathcal{V}} \left\{ E[A_{\theta_0}^v(W_i) A_{\theta_0}^{v'}(W_i)] \Gamma_0(v) \Gamma_0(v')' \right\} d\mu(v) d\mu(v') \nonumber \\
&=& \int \int_{\mathcal{V}\times \mathcal{V}} \left\{ \Psi(v,v') \Gamma_0(v) \Gamma_0(v')' \right\} d\mu(v) d\mu(v'), \nonumber
\end{eqnarray}
where all elements of $E[\xi \xi']$ are finite.
Hence, $\sqrt{n}(\tilde{\theta}_n-\theta_0) \rightsquigarrow N(0,\Sigma_0)$ by (\ref{theta_tilde}).
Consequently, $\tilde{\theta}_n=\theta_0 + O_p(n^{-1/2})$, and $\{\delta_n\}$ can be assumed to satisfy $P(\|\tilde{\theta}_n - \theta_0\| \geq \delta_n)\rightarrow 0$.
Because $\theta_0$ is an interior point of $\Theta$, $\tilde{\theta}_n$ lies in $\Theta$ with probability approaching one.
To simplify the argument, we assume that $\|\tilde{\theta}_n-\theta_0\|<\delta_n$ and $\tilde{\theta}_n$ always belongs to $\Theta$.

Because $|D_{\theta}(v)| \leq |\Gamma_0(v)'(\theta-\theta_0)| + o(\|\theta-\theta_0\|)$ uniformly over $v$ by Lemma 3, we have
$$
\|D(\tilde{\theta}_n)\|_{\mu} \leq \|\Gamma_0(v)'(\tilde{\theta}_n-\theta_0)\|_{\mu}  + o_p(\|\tilde{\theta}_n-\theta_0\|) = O_p(n^{-1/2}).
$$
By the triangle inequality and Lemma 4, we have $\|\hat{D}_{n,\tilde{\theta}_n}\|_{\mu}-\|D_{\tilde{\theta}_n}\|_{\mu}-\|\hat{D}_{n,\theta_0}\|_{\mu} = o_p(n^{-1/2})$, and hence $\|\hat{D}_{n,\tilde{\theta}_n}\|_{\mu} = O_p(n^{-1/2})$.
Then, we can argue as for $\hat{\theta}_n$ to deduce that
$$
\|\hat{D}_{n,\tilde{\theta}_n}-\hat{L}_{n,\tilde{\theta}_n}\|_{\mu} = o_p(n^{-1/2}).
$$

Above, we showed that $\|\hat{D}_{n,\hat{\theta}_n}-\hat{L}_{n,\hat{\theta}_n}\|_{\mu} = o_p(n^{-1/2})$ and $\|\hat{D}_{n,\tilde{\theta}_n}-\hat{L}_{n,\tilde{\theta}_n}\|_{\mu} = o_p(n^{-1/2})$.
Therefore, we have
\begin{eqnarray}
\|\hat{L}_{n,\hat{\theta}_n}\|_{\mu} - o_p(n^{-1/2}) &\leq & \|\hat{D}_{n,\hat{\theta}_n}\|_{\mu} \nonumber \\
&\leq & \|\hat{D}_{n,\tilde{\theta}_n}\|_{\mu} + o_p(n^{-1/2}) \nonumber \\
&\leq & \|\hat{L}_{n,\tilde{\theta}_n}\|_{\mu}  + o_p(n^{-1/2}). \nonumber
\end{eqnarray}
That is,
$$
\|\hat{L}_{n,\hat{\theta}_n}\|_{\mu} = \|\hat{L}_{n,\tilde{\theta}_n}\|_{\mu}  + o_p(n^{-1/2}),
$$
and by squaring both sides, we have
$$
\|\hat{L}_{n,\hat{\theta}_n}\|_{\mu}^2 = \|\hat{L}_{n,\tilde{\theta}_n}\|_{\mu}^2  + o_p(n^{-1}),
$$
where the cross product term is absorbed into $o_p(n^{-1})$ because $\|\hat{L}_{n,\tilde{\theta}_n}\|_{\mu} = O_p(n^{-1/2})$.
Because $\hat{L}_{n,\tilde{\theta}_n}(\cdot)$ and $\Gamma_0(\cdot)$ are orthogonal according to the definition of $\tilde{\theta}_n$, we can obtain
\begin{eqnarray}
\|\hat{L}_{n,\theta}\|_{\mu}^2 &=& \|\hat{L}_{n,\tilde{\theta}_n}(v) + \Gamma_0(v)' (\theta-\tilde{\theta}_n) \|_{\mu}^2 \nonumber \\
&=& \|\hat{L}_{n,\tilde{\theta}_n}\|_{\mu}^2 + \|\Gamma_0(v)' (\theta-\tilde{\theta}_n)\|_{\mu}^2. \nonumber 
\end{eqnarray}
By making $\theta$ equal to $\hat{\theta}_n$, we have
$$
o_p(n^{-1}) = \|\Gamma_0(v) (\hat{\theta}_n-\tilde{\theta}_n)\|_{\mu}^2 \geq c^2 \|\hat{\theta}_n-\tilde{\theta}_n\|^2.
$$
Hence, $\sqrt{n}(\hat{\theta}_n-\theta_0)=\sqrt{n}(\tilde{\theta}_n-\theta_0)+o_p(1)\rightsquigarrow N(0,\Delta_0^{-1} \Sigma_0 \Delta_0^{-1})$.
\end{proof}
\vspace{0.1in}

Suppose that each $W_i$ is a coordinate function of $(\prod_{i=1}^{\infty}S, \prod_{i=1}^{\infty}\sigma(S), \prod_{i=1}^{\infty} P)$.
Let $\omega$ denote one of the realizations of $W_i$, and let $(W_{n1}^*, \cdots , W_{nn}^*)$ denote the bootstrap sample.
Following \cite{hahn1996note}, we introduce the following notations.
Let $\{\zeta^*_n\}$ be a sequence of some bootstrap statistic: each $\zeta^*_n$ is some function $f_n(W_{n1}^*, \cdots , W_{nn}^*)$ of the bootstrap sample.
We write $\zeta^*_n = O_p^{\omega}(a_n)$ if $\zeta^*_n$, when conditioned on $\omega$, is $O_p(a_n)$ for almost all $\omega$.
If $\zeta^*_n$, when conditioned on $\omega$, is $o_p(a_n)$ for almost all $\omega$, we write $\zeta^*_n = o_p^{\omega}(a_n)$.
We write $\zeta^*_n = O_B(1)$ if, for a given subsequence $\{n'\}$, there exists a further subsequence $\{n''\}$ such that $\zeta^*_{n''} = O_p^{\omega}(1)$.
If for any subsequence $\{n'\}$ there is a further subsequence $\{n''\}$ such that $\zeta^*_{n''} = o_p^{\omega}(1)$, we write $\zeta^*_n = o_B(1)$.
Note that $\zeta^*_n = o_B(1)$ if and only if $\zeta^*_n$ converges weakly to zero in probability.

\begin{proof}[Proof of Theorem 4]
The proof is similar to that of \cite{brown2002weighted}.
First, we define
\begin{eqnarray}
M(\theta) &\equiv & \int D_{\theta}(v)^2 d\mu(v), \nonumber \\
M_n(\theta) &\equiv & \int \hat{D}_{n,\theta}(v)^2 d\mu(v), \nonumber \\
M_n^*(\theta) &\equiv & \int \hat{D}_{n,\theta}^*(v)^2 d\mu(v). \nonumber
\end{eqnarray}
Then, for any $\theta \rightarrow \theta_0$, we have
\begin{eqnarray}
M_n^*(\theta)-M_n(\theta) &=& \int (\hat{D}^*_{n,\theta}(v)-\hat{D}_{n,\theta}(v))^2 d\mu(v) + 2 \int \hat{D}_{n,\theta}(v) (\hat{D}^*_{n,\theta}(v)-\hat{D}_{n,\theta}(v)) d\mu(v) \nonumber \\
&=& \int (\hat{D}^*_{n,\theta_0}(v)-\hat{D}_{n,\theta_0}(v))^2 d\mu(v) + \left[ \int (\hat{D}^*_{n,\theta}(v)-\hat{D}_{n,\theta}(v))^2 d\mu(v) \right. \nonumber \\
& & \left. - \int (\hat{D}^*_{n,\theta_0}(v)-\hat{D}_{n,\theta_0}(v))^2 d\mu(v) \right] + 2 \int D_{\theta}(v) (\hat{D}^*_{n,\theta}(v)-\hat{D}_{n,\theta}(v)) d\mu(v) \nonumber \\
& & + 2 \int (\hat{D}_{n,\theta}(v) - D_{\theta}(v)) (\hat{D}^*_{n,\theta}(v)-\hat{D}_{n,\theta}(v)) d\mu(v). \nonumber 
\end{eqnarray}
Suppose that $\| \theta - \theta_0 \| \leq \delta_n$ for $\delta_n \downarrow 0$.
By Lemma 6, we obtain $\|\sqrt{n}(\hat{D}^*_{n,\theta}-\hat{D}_{n,\theta})\|_{\mu} - \|\sqrt{n}(\hat{D}^*_{n,\theta_0}-\hat{D}_{n,\theta_0})\|_{\mu}=o_B(1)$.
Hence, 
$$
\int (\hat{D}^*_{n,\theta}(v)-\hat{D}_{n,\theta}(v))^2 d\mu(v) - \int (\hat{D}^*_{n,\theta_0}(v)-\hat{D}_{n,\theta_0}(v))^2 d\mu(v) =o_B(n^{-1}).
$$
Similarly, by the Donsker property of $\{A_{\theta}^v:\theta \in \Theta, v \in \mathcal{V}\}$,
\begin{eqnarray}
& & \int (\hat{D}_{n,\theta}(v) - D_{\theta}(v)) (\hat{D}^*_{n,\theta}(v)-\hat{D}_{n,\theta}(v)) d\mu(v) \nonumber \\
&=& \int \hat{D}_{n,\theta_0}(v) (\hat{D}^*_{n,\theta}(v)-\hat{D}_{n,\theta}(v)) d\mu(v) + o_p(n^{-1}). \nonumber
\end{eqnarray}
Therefore, we have
\begin{eqnarray}
M_n^*(\theta)-M_n(\theta) &=& \int (\hat{D}^*_{n,\theta_0}(v)-\hat{D}_{n,\theta_0}(v))^2 d\mu(v) + 2 \int \hat{D}_{n,\theta_0}(v) (\hat{D}^*_{n,\theta}(v)-\hat{D}_{n,\theta}(v)) d\mu(v) \nonumber \\
& & + 2 \int D_{\theta}(v) (\hat{D}^*_{n,\theta}(v)-\hat{D}_{n,\theta}(v)) d\mu(v) + o_B(n^{-1}) \nonumber \\
&=& \int (\hat{D}^*_{n,\theta_0}(v)-\hat{D}_{n,\theta_0}(v))^2 d\mu(v) + 2 \int \hat{D}_{n,\theta_0}(v) (\hat{D}^*_{n,\theta}(v)-\hat{D}_{n,\theta}(v)) d\mu(v) \nonumber \\
& & + 2 (\theta - \theta_0)' \int \Gamma_0(v) (\hat{D}^*_{n,\theta}(v)-\hat{D}_{n,\theta}(v)) d\mu(v) \nonumber \\
& & + o_B(n^{-1/2}\| \theta-\theta_0 \| + n^{-1}). \nonumber
\end{eqnarray}
Consequently, for $\theta \rightarrow \theta_0$ and $\eta \rightarrow \theta_0$,
\begin{eqnarray}
& & M_n^*(\theta)-M_n^*(\eta) \nonumber \\
&=& [(M_n^*-M_n)(\theta)-(M_n^*-M_n)(\eta)] + [(M_n-M)(\theta)-(M_n-M)(\eta)] + [M(\theta)-M(\eta)] \nonumber \\
&=& 2 (\theta - \eta)' \int \Gamma_0(v) \left[ (\hat{D}^*_{n,\theta_0}(v)-\hat{D}_{n,\theta_0}(v))-(\hat{D}_{n,\theta_0}(v)-D_{\theta_0}(v)) \right] d\mu(v)  \nonumber \\
& & + \int \left[ (\theta-\theta_0)' \Gamma_0(v)+o(\|\theta-\theta_0\|) \right]^2 d\mu(v) - \int \left[ (\eta-\theta_0)' \Gamma_0(v)+o(\|\eta-\theta_0\|) \right]^2 d\mu(v) \nonumber \\
& & +o_B(n^{-1/2}\|\theta-\theta_0\|+n^{-1/2}\|\eta-\theta_0\|+n^{-1}) \nonumber \\
&=& 2 (\theta - \eta)' \int \Gamma_0(v) \left[ (\hat{D}^*_{n,\theta_0}(v)-\hat{D}_{n,\theta_0}(v))-(\hat{D}_{n,\theta_0}(v)-D_{\theta_0}(v)) \right] d\mu(v)  \nonumber \\
& & + (\theta-\theta_0)'\Delta_0(\theta-\theta_0) - (\eta-\theta_0)'\Delta_0(\eta-\theta_0) \nonumber \\
& & +o_B(\|\theta-\theta_0\|^2 + \|\eta-\theta_0\|^2 +n^{-1/2}\|\theta-\theta_0\| + n^{-1/2}\|\eta-\theta_0\| + n^{-1}). \label{M_n^*-M_n^*}
\end{eqnarray}
We define 
\begin{eqnarray}
\gamma_n &\equiv & \Delta_0^{-1} \int \Gamma_0(v)(\hat{D}_{n,\theta_0}(v)-D_{\theta_0}(v)) d\mu(v), \nonumber \\
\gamma_n^* &\equiv & \Delta_0^{-1} \int \Gamma_0(v) (\hat{D}^*_{n,\theta_0}(v)-\hat{D}_{n,\theta_0}(v))d\mu(v). \nonumber
\end{eqnarray}
Then, we can rewrite (\ref{M_n^*-M_n^*}) by
\begin{eqnarray}
M_n^*(\theta)-M_n^*(\eta) &=& 2(\theta-\eta)' \Delta_0 (\gamma_n+\gamma_n^*) + (\theta-\theta_0)'\Delta_0(\theta-\theta_0) - (\eta-\theta_0)'\Delta_0(\eta-\theta_0) \nonumber \\
& & +o_B(\|\theta-\theta_0\|^2 + \|\eta-\theta_0\|^2 +n^{-1/2}\|\theta-\theta_0\| + n^{-1/2}\|\eta-\theta_0\| + n^{-1}). \nonumber
\end{eqnarray}
We take $\theta = \hat{\theta}_n^*$ and $\eta = \theta_0-(\gamma_n+\gamma_n^*)$.
Observe that $\eta \in \Theta$ for $n$ is sufficiently large, because $\theta_0$ is an interior point of $\Theta$.
Hence, we have
\begin{eqnarray}
0 &\geq & M_n^*(\hat{\theta}_n^*) - M_n^*(\theta_0-(\gamma_n+\gamma_n^*)) \nonumber \\
&=& 2\left[ (\hat{\theta}_n^*-\theta_0) + (\gamma_n + \gamma_n^*) \right]' \Delta_0 (\gamma_n + \gamma_n^*) \nonumber \\
& & + (\hat{\theta}_n^*-\theta_0)' \Delta_0 (\hat{\theta}_n^*-\theta_0) - (\gamma_n + \gamma_n^*)' \Delta_0 (\gamma_n + \gamma_n^*) \nonumber \\
& & + o_B(\|\hat{\theta}_n^*-\theta_0\|^2 + \|\gamma_n + \gamma_n^*\|^2 +n^{-1/2}\|\hat{\theta}_n^*-\theta_0\| + n^{-1/2}\|\gamma_n + \gamma_n^*\| + n^{-1}) \nonumber \\
&=& \left[ (\hat{\theta}_n^*-\theta_0) + (\gamma_n + \gamma_n^*) \right]' \Delta_0 \left[ (\hat{\theta}_n^*-\theta_0) + (\gamma_n + \gamma_n^*) \right] \nonumber \\
& & + o_B(\|\hat{\theta}_n^*-\theta_0\|^2 + \|\gamma_n + \gamma_n^*\|^2 +n^{-1/2}\|\hat{\theta}_n^*-\theta_0\| + n^{-1/2}\|\gamma_n + \gamma_n^*\| + n^{-1}). \nonumber 
\end{eqnarray}
By the same argument in Theorem 4, we have $\|\hat{\theta}_n^* - \hat{\theta}_n\| = O_B(n^{-1/2})$.
Hence, $\|\hat{\theta}_n^*-\theta_0\| \leq \|\hat{\theta}_n^* - \hat{\theta}_n\| + \|\hat{\theta}_n - \theta_0\| = O_B(n^{-1/2}) + O_P(n^{-1/2})$.
Since $\gamma_n = O_p(n^{-1/2})$ and $\gamma_n^* =O_B(n^{-1/2})$, we have
$$
n \| \hat{\theta}_n^*-\theta_0 + (\gamma_n + \gamma_n^*) \|^2 = o_B(1).
$$
Because it follows from Theorem 4 that
$$
\hat{\theta}_n - \theta_0 = -\gamma_n + o_p(n^{-1/2}),
$$
and we can obtain $\hat{\theta}_n^*-\hat{\theta}_n= -\gamma_n^* + o_B(n^{-1/2})$.
The term $\gamma_n^*$ has the same limiting distribution as $\gamma_n$ by the bootstrap theorem for the mean in $\mathbb{R}^{d_{\theta}}$.
This concludes the proof.
\end{proof}
\vspace{0.1in}

\begin{proof}[Proof of Theorem 5]
We establish the partial identification of $g_t$ by showing that we can identify functions $T^U_{t,x',x}:\mathbb{R}\rightarrow \mathbb{R}$ and $T^L_{t,x',x}:\mathbb{R}\rightarrow \mathbb{R}$ that satisfy
\begin{eqnarray}
g_t(x',u) &\leq & T^U_{t,x',x}\left( g_t(x,u) \right), \nonumber \\
g_t(x',u) &\geq & T^L_{t,x',x}\left( g_t(x,u) \right). \label{TU and TL*}
\end{eqnarray}
If $T^U_{x',x}$ is identified for all $x,x'$, then we can obtain a lower bound of the function $g$ as follows.
For any random variables $V,W$, we define
\begin{eqnarray}
F^{+}_{V|W}(v|w) &\equiv & P\left( V \leq v | W=w \right), \nonumber \\
F^{-}_{V|W}(v|w) &\equiv & P\left( V < v | W=w \right), \nonumber
\end{eqnarray}
where $F^+$ is an usual distribution function.
In addition, we define 
\begin{equation}
G_{t,x}^L(y) \equiv \int F^+_{Y_t|X_t=x'}\left( T^U_{t,x',x}(y) \right) dF_{X_t}(x'). \label{def GL}
\end{equation}
Then, we have
\begin{eqnarray}
G_{t,x}^L\left( g_t(x,u) \right) &=& \int F^+_{Y_t|X_t=x'}\left( T^U_{t,x',x}\left( g_t(x,u) \right) \right) dF_{X_t}(x') \nonumber \\
& \geq & \int F^+_{Y_t|X_t=x'}\left(  g_t(x',u) \right) dF_{X_t}(x') \nonumber \\
&=& \int P\left( g_t(x',U_t) \leq g_t(x',u) | X_t = x' \right) dF_{X_t}(x') \nonumber \\
& \geq & \int P\left( U_t \leq u | X_t = x' \right) dF_{X_t}(x') = u, \label{GL*}
\end{eqnarray}
where the first inequality follows from (\ref{TU and TL*}).
Because $g_t(x,u)$ is weakly increasing in $u$, we have $\{U_t\leq u\} \subset \{g_t(x,U_t) \leq g_t(x,u)\}$ and the second inequality of (\ref{GL*}) holds.
Hence, because $G_{t,x}^L\left( g_t(x,u) \right) \geq u$, we can obtain a lower bound
\begin{equation}
g_t(x,u) \geq \inf \{y \in [\underline{y},\overline{y}] :G_{t,x}^L(y) \geq u\}. \label{lower bound*}
\end{equation}
Similarly, we define
\begin{equation}
G^U_{t,x}(y) \equiv \int F^-_{Y_t|X_t=x'}\left( T^L_{t,x',x}(y) \right) dF_{X_t}(x'). \label{def GU}
\end{equation}
Then, we have
\begin{eqnarray}
G^U_{t,x}\left( g_t(x,u) \right) &= & \int F^-_{Y_t|X_t=x'}\left(  T^L_{t,x',x}(g_t(x,u)) \right) dF_{X_t}(x') \nonumber \\
&\leq & \int F^-_{Y_t|X_t=x'}\left(  g_t(x',u) \right) dF_{X_t}(x') \nonumber \\
&=& \int P\left( g_t(x',U_t) < g_t(x',u) | X_t = x' \right) dF_{X_t}(x') \nonumber \\
& \leq & \int P\left( U_t < u | X_t = x' \right) dF_{X_t}(x') = u. \label{GU*}
\end{eqnarray}
Owing to the weak monotonicity of $g_t$, we have $\{g_t(x,U_t) < g_t(x,u)\} \subset \{U_t < u\}$ and the second inequality of (\ref{GU*}) holds.
Hence, similarly, we can obtain an upper bound
\begin{equation}
g_t(x,u) \geq \sup \{y\in [\underline{y},\overline{y}]: G^U_{t,x}(y) \leq u\}. \label{upper bound*}
\end{equation}

We here describe the construction of the functions $T_{t,x',x}^U(y)$ and $T_{t,x',x}^L(y)$ that satisfy (\ref{TU and TL*}).
We define 
\begin{eqnarray}
Q^+_{Y_t|\mathbf{X}}(\tau | \mathbf{x}) &\equiv & \sup \{y \in [\underline{y},\overline{y}] : F^-_{Y_t|\mathbf{X}}(y|\mathbf{x}) \leq \tau \}, \nonumber \\
Q^-_{Y_t|\mathbf{X}}(\tau | \mathbf{x}) &\equiv & \inf \{y \in [\underline{y},\overline{y}] : F^+_{Y_t|\mathbf{X}}(y|\mathbf{x}) \geq \tau \}. \nonumber
\end{eqnarray}
Because $\{U:U\leq u\} \subset \{U:g(x,U) \leq g(x,u)\}$ and $\{U:g(x,U) < g(x,u)\} \subset \{U:U < u\}$, for all $(x_1,x_2) \in \mathcal{X}_{12}$ and $t,s \in \{(1,2),(2,1)\}$, we have
\begin{eqnarray}
F^-_{Y_t|\mathbf{X}}\left( g_t(x_t,u) | x_1,x_2 \right) &=& P\left( g_t(x_t,U_t) < g_t(x_t,u) | X_1=x_1, X_2=x_2 \right) \nonumber \\
&\leq & P\left( U_t < u | X_1=x_1, X_2=x_2 \right) \nonumber \\
&=& P\left( U_s < u | X_1=x_1, X_2=x_2 \right) \nonumber \\
&\leq & P\left( g_s(x_s,U_s) < g_s(x_s,u) | X_1=x_1, X_2=x_2 \right) \nonumber \\
&=& F^+_{Y_s|\mathbf{X}}\left( g_s(x_s,u) | x_1,x_2 \right). \nonumber
\end{eqnarray}
For $t\neq s$, we define
\begin{eqnarray}
\tilde{T}^{U,t,s}_{x_t,x_{s}}(y) &\equiv & Q^+_{Y_t|(X_t,X_{s})=(x_t,x_{s})}\left( F^+_{Y_{s}|(X_t,X_{s})=(x_t,x_{s})}( y ) \right) \nonumber \\
\tilde{T}^{L,t,s}_{x_t,x_{s}}(y) &\equiv & Q^-_{Y_t|(X_t,X_{s})=(x_t,x_{s})}\left( F^-_{Y_{s}|(X_t,X_{s})=(x_t,x_{s})}( y ) \right). \nonumber
\end{eqnarray}
Then, we have
$$
g_t(x_t,u) \leq \tilde{T}^{U,t,s}_{x_t,x_{s}}\left( g_s(x_s,u) \right),
$$
because $Q^+_{Y_t|\mathbf{X}}(F^-_{Y_t|\mathbf{X}}(y | \mathbf{x}) | \mathbf{x} ) = \sup \{y' \in [\underline{y},\overline{y}]: F^-_{Y_t|\mathbf{X}}(y' | \mathbf{x}) \leq F^+_{Y_t|\mathbf{X}}(y | \mathbf{x}) \} \geq y$. 
Hence, if $(x',\tilde{x})\in supp(X_t,X_s)$ and $(x,\tilde{x})\in supp(X_t,X_s)$, then we have
\begin{equation}
g_t(x',u) \leq \tilde{T}^{U,t,s}_{x',\tilde{x}} \circ \tilde{T}^{U,s,t}_{\tilde{x},x}\left( g_t(x,u) \right). \label{tilde TU*}
\end{equation}
Similarly, we have
\begin{equation}
g_t(x',u) \geq \tilde{T}^{L,t,s}_{x',\tilde{x}} \circ \tilde{T}^{L,s,t}_{\tilde{x},x}\left( g_t(x,u) \right). \label{tilde TL*}
\end{equation}
We define
\begin{eqnarray}
T^U_{t,x',x}(y) &\equiv & \begin{cases}
\inf_{\tilde{x}} \{ \tilde{T}^{U,t,s}_{x',\tilde{x}} \circ \tilde{T}^{U,s,t}_{\tilde{x},x}(y) \} & \text{if $x \neq x'$} \\
y & \text{if $x = x'$}
\end{cases} \nonumber \\
T^L_{t,x',x}(y) &\equiv & \begin{cases}
\sup_{\tilde{x}} \{ \tilde{T}^{L,t,s}_{x',\tilde{x}} \circ \tilde{T}^{L,s,t}_{\tilde{x},x}(y)\} & \text{if $x \neq x'$} \\
y & \text{if $x = x'$}
\end{cases}. \nonumber
\end{eqnarray}
Then, these functions satisfy (\ref{TU and TL*}).
\end{proof}
\vspace{0.1in}

\begin{proof}[Corollary 3]
Similarly to (\ref{tilde TU*}) and (\ref{tilde TL*}), for $t \neq s$, we have
\begin{eqnarray}
g_t(x',u) & \leq & \tilde{T}^{U,t,s}_{x',\tilde{x}} \circ \tilde{T}^{U,s,t}_{\tilde{x},x} \left( g_t(x,u) \right), \nonumber \\
g_t(x',u) & \geq & \tilde{T}^{L,t,s}_{x',\tilde{x}} \circ \tilde{T}^{L,s,t}_{\tilde{x},x} \left( g_t(x,u) \right), \nonumber
\end{eqnarray}
where 
\begin{eqnarray}
\tilde{T}^{U,t,s}_{x_t,x_s}(y) & \equiv & \inf_{\mathbf{x}_{-(t,s)}} Q^+_{Y_t|(X_t,X_s,\mathbf{X}_{-(t,s)})=(x_t,x_s,\mathbf{x}_{-(t,s)})}\left( F^+_{Y_s|(X_t,X_s,\mathbf{X}_{-(t,s)})=(x_t,x_s,\mathbf{x}_{-(t,s)})}( y ) \right), \nonumber \\
\tilde{T}^{L,t,s}_{x_t,x_s}(y) & \equiv & \sup_{\mathbf{x}_{-(t,s)}} Q^-_{Y_t|(X_t,X_s,\mathbf{X}_{-(t,s)})=(x_t,x_s,\mathbf{x}_{-(t,s)})}\left( F^-_{Y_s|(X_t,X_s,\mathbf{X}_{-(t,s)})=(x_t,x_s,\mathbf{x}_{-(t,s)})}( y ) \right), \nonumber
\end{eqnarray}
and $\mathbf{X}_{-(t,s)}$ denotes a vector of $\mathbf{X}$ except $X_t$ and $X_s$.
Hence,
\begin{eqnarray}
\hat{T}^U_{t,x',x}(y) & \equiv & \begin{cases}
\inf_{s \neq t, \tilde{x}\in \mathcal{X}_s} \{ \tilde{T}^{U,t,s}_{x',\tilde{x}} \circ \tilde{T}^{U,s,t}_{\tilde{x},x}(y) \} & \text{if $x \neq x'$} \\
y & \text{if $x = x'$}
\end{cases} \nonumber
\end{eqnarray}
and
\begin{eqnarray}
\hat{T}^L_{t,x',x}(y) & \equiv & \begin{cases}
\sup_{s \neq t, \tilde{x}\in \mathcal{X}_s} \{ \tilde{T}^{L,t,s}_{x',\tilde{x}} \circ \tilde{T}^{L,s,t}_{\tilde{x},x}(y) \} & \text{if $x \neq x'$} \\
y & \text{if $x = x'$}
\end{cases} \nonumber
\end{eqnarray}
satisfy inequality (\ref{TU and TL*}).
Define
\begin{eqnarray}
\hat{G}^L_{t,x}(y) &\equiv & \int F^+_{Y_t|X_t=x'}\left( \hat{T}^U_{t,x',x}(y) \right) dF_{X_t}(x'), \nonumber \\
\hat{G}^U_{t,x}(y) &\equiv & \int F^-_{Y_t|X_t=x'}\left( \hat{T}^L_{t,x',x}(y) \right) dF_{X_t}(x'). \label{GL and GU T>=3}
\end{eqnarray}
By a similar argument to the proof for Theorem 2, we have $g_t(x,u) \geq \inf \{y \in [\underline{y},\overline{y}] : G_{t,x}^L(y) \geq u\}$ and $g_t(x,u) \leq \sup \{y \in [\underline{y},\overline{y}] : G_{t,x}^U(y) \leq u\}$.
\end{proof}
\vspace{0.1in}

\begin{proof}[Proof of Theorems 6 and 7]
We define $\mu_{T}$ as a probability measure on $\{1,\cdots ,T\}$ such that $\mu_T(\{t\})=1/T$ for all $t \in \{1, \cdots , T\}$. 
Let $\tilde{\mu}$ be $\mu_T \times \mu$, then $\frac{1}{T} \sum_{t=1}^T \|\hat{D}^t_{n,\theta}\|_{\mu}^2 = \frac{1}{T} \sum_{t=1}^T \int \hat{D}^t_{n,\theta}(v)^2 d\mu(v) = \|\hat{D}^t_{n,\theta}(v)\|_{\tilde{\mu}}^2$, where $\| \cdot \|_{\tilde{\mu}}$ is the $L_2$-norm with respect to $\tilde{\mu}$.
Hence, the estimator $\hat{\theta}_n$ satisfies
$$
\| \hat{D}^t_{n,\hat{\theta}_n}(v) \|_{\tilde{\mu}} = \inf_{\theta \in \Theta} \| \hat{D}^t_{n,\theta}(v) \|_{\tilde{\mu}}.
$$
Therefore, we can prove Theorems 6 and 7 by following arguments similar to the proofs of Theorems 3 and 4, respectively.
\end{proof}

\section*{Appendix 3: Auxiliary Lemmas}

\begin{Lemma}
Under Assumptions C3 and C4, $\|D_\theta\|_{\mu}$ is continuous in $\theta$.
\end{Lemma}

\begin{proof}
By Assumption C4, the density $f_{Y_t|\mathbf{X}}(y|\mathbf{x})$ is bounded above by a constant $C$.
For any $\theta',\theta$ and $v$, we have
\begin{eqnarray}
& & |D_{\theta'}(v)-D_{\theta}(v)| \nonumber \\
&\leq & 2 \max_t |E\left[ \left( \mathbf{1}\{Y_t\leq g_t(X_t,v_u;\theta')\}-\mathbf{1}\{Y_t\leq g_t(X_t,v_u;\theta)\}  \right) \mathbf{1}\{\mathbf{X}\leq v_\mathbf{x}\} \right]| \nonumber \\
&\leq & 2 \max_t |E\left[ \left( F_{Y_t|\mathbf{X}} \left( g_t(X_t,v_u;\theta')| \mathbf{X} \right) - F_{Y_t|\mathbf{X}} \left( g_t(X_t,v_u;\theta)| \mathbf{X} \right) \right) \mathbf{1}\{\mathbf{X}\leq v_\mathbf{x}\} \right]| \nonumber \\
&\leq & 2 \max_t  \int \left| \int_{g_t(x_t,v_u;\theta)}^{g_t(x_t,v_u;\theta')} f_{Y_t|\mathbf{X}}(y|\mathbf{x}) dy  \right| dF_{\mathbf{X}}(\mathbf{x})  \nonumber \\
&\leq & 2 C \max_t \int \left| g_t(x_t,v_u;\theta')-g_t(x_t,v_u;\theta) \right| dF_{X_t}(x_t) \leq 2CK \|\theta'-\theta\|. \nonumber 
\end{eqnarray}
Hence, $| \|D_{\theta'}\|_{\mu} - \|D_{\theta}\|_{\mu} | \leq \|D_{\theta'}-D_{\theta}\|_{\mu} \leq 2CK \|\theta'-\theta\|$, which implies the continuity of $\|D_{\theta}\|_{\mu}$.
\end{proof}
\vspace{0.1in}

\begin{Lemma}
Under Assumptions C3, C4, and C5, 
\begin{equation}
\sup_{\theta,v} \left| \hat{D}_{n,\theta}(v) - D_{\theta}(v) \right| = o_{a.s.}(1), \label{Glivenko-Cantelli}
\end{equation}
and for any $\delta_n \downarrow 0$
\begin{equation}
\sup_{v \in \mathcal{V}, \|\theta-\theta_0\| < \delta_n} \left| \sqrt{n} \left( \hat{D}_{n,\theta}(v) - D_{\theta}(v) \right) - \sqrt{n} \left( \hat{D}_{n,\theta_0}(v) - D_{\theta_0}(v) \right) \right| = o_p(1). \label{Donsker}
\end{equation}
\end{Lemma}

\begin{proof}
The collection of indicator functions $\{ \mathbf{x} \mapsto \mathbf{1}\{\mathbf{x} \leq v_\mathbf{x}\} : v_\mathbf{x} \in \mathcal{X}_{12} \}$ is a VC-class.
By Assumption C5, the collection of indicator functions for subgraphs of $\{g_t(\cdot,v_u;\theta) : \theta \in \Theta, v_u \in \mathcal{U} \}$ is also a VC-class.
By reference to Examples 2.10.7 and 2.10.8 in \cite{van1996weak}, $\{A_{\theta}^v: \theta \in \Theta, v \in \mathcal{V}\}$ is $P$-Donsker, and also $P$-Glivenko--Cantelli.
Hence, we have (\ref{Glivenko-Cantelli}) and for any $\delta_n \downarrow 0$
\begin{eqnarray}
& & \sup_{(\theta,v), (\theta',v') : \mathbb{P}(A_{\theta}^v - A_{\theta'}^{v'})^2 < \delta_n} | \mathbb{G}_n A_{\theta}^v - \mathbb{G}_n A_{\theta'}^{v'} | = o_p(1) \nonumber \\
& \Rightarrow & \sup_{v, \theta, \theta_0 : \mathbb{P}(A_{\theta}^v - A_{\theta_0}^v)^2 < \delta_n} | \mathbb{G}_n A_{\theta}^v - \mathbb{G}_n A_{\theta_0}^v | = o_p(1). \nonumber
\end{eqnarray}
Then, we have
\begin{eqnarray}
& & \mathbb{P}(A_{\theta}^v - A_{\theta_0}^v)^2 \nonumber \\
&\leq & E\left[ \left\{ (\mathbf{1}\{Y_1 \leq g_1(X_1,v_u;\theta)\} - \mathbf{1}\{Y_1 \leq g_1(X_1,v_u;\theta_0)\}) \right. \right. \nonumber \\
& & \left. \left. - (\mathbf{1}\{Y_2 \leq g_2(X_2,v_u;\theta)\} - \mathbf{1}\{Y_2 \leq g_2(X_2,v_u;\theta_0)\}) \right\}^2 \right] \nonumber \\
&\leq & 4 \max_{t} E\left[ \left| \mathbf{1}\{Y_t \leq g_t(X_t,v_u;\theta)\} - \mathbf{1}\{Y_t \leq g_t(X_t,v_u;\theta_0)\} \right| \right] \nonumber \\
&=& 4 \max_{t} E\left[ \left| \mathbf{1}\{g_t(X_t,v_u;\theta_0) < Y_t \leq g_t(X_t,v_u;\theta)\} + \mathbf{1}\{g_t(X_t,v_u;\theta)\} < Y_t \leq g_t(X_t,v_u;\theta_0)\} \right| \right] \nonumber \\
&\leq & 8 \max_{t} \int \left| F_{Y_t|\mathbf{X}} \left( g_t(x_t,v_u;\theta)| \mathbf{x} \right) - F_{Y_t|\mathbf{X}} \left( g_t(x_t,v_u;\theta_0)| \mathbf{x} \right) \right| dF_{\mathbf{X}}(\mathbf{x}) \nonumber \\
&\leq & 8 CK \|\theta - \theta_0\|. \nonumber
\end{eqnarray}
Because $\|\theta - \theta_0\|\rightarrow 0$ implies that $\mathbb{P}(A_{\theta}^v - A_{\theta_0}^v)^2 \rightarrow 0$, we have (\ref{Donsker}).
\end{proof}
\vspace{0.1in}

\begin{Lemma}
Under Assumptions C4, N2, and N3, $D_{\theta}(v)$ is continuously differentiable in $\theta$ in a neighborhood of $\theta_0$ for all $v$, and $|D_{\theta}(v) - D_{\theta_0}(v) - \Gamma_0(v)'(\theta - \theta_0)| = o(\|\theta - \theta_0\|)$ uniformly over $v$.
\end{Lemma}

\begin{proof}
First, we show continuous differentiability of $D_{\theta}(v)$.
For all $v$ and $\theta$ in the neighborhood,
\begin{eqnarray}
\nabla_{\theta} D_{\theta}(v) &=& \nabla_{\theta} E\left[ \left(F_{Y_1|\mathbf{X}}(g_1(X_1,v_u;\theta)|\mathbf{X}) - F_{Y_2|\mathbf{X}}(g_2(X_2,v_u;\theta)|\mathbf{X}) \right) \mathbf{1}\{\mathbf{X} \leq v_{\mathbf{x}}\} \right] \nonumber \\
&=& \nabla_{\theta} \int_{\{\mathbf{x} \leq v_{\mathbf{x}}\}} \left(F_{Y_1|\mathbf{X}}(g_1(x_1,v_u;\theta)|\mathbf{x}) - F_{Y_2|\mathbf{X}}(g_2(x_2,v_u;\theta)|\mathbf{x}) \right) dF_{\mathbf{X}}(\mathbf{x}). \nonumber 
\end{eqnarray}
Let $C$ be a constant such that $f_{Y_t|\mathbf{X}}(y|\mathbf{x}) \leq C$.
Because $|f_{Y_t|\mathbf{X}}(g_t(x_t,v_u;\theta)|\mathbf{x}) \nabla_{\theta} g_t(x_t,v_u;\theta)|$ is bounded by the integrable function $C \nabla \bar{g}(x_t)$, we can interchange a differential operator with an integral.
Hence, we have
\begin{eqnarray}
\nabla_{\theta} D_{\theta}(v) &=& E[f_{Y_1|\mathbf{X}}(g_1(X_1,v_u;\theta)|\mathbf{X}) \nabla_{\theta} g_1(X_1,v_u;\theta) \mathbf{1}\{\mathbf{X}\leq v_\mathbf{x}\}] \nonumber \\
& & - E[f_{Y_2|\mathbf{X}}(g_2(X_2,v_u;\theta)|\mathbf{X}) \nabla_{\theta} g_2(X_2,v_u;\theta)\mathbf{1}\{\mathbf{X}\leq v_\mathbf{x}\}]. \nonumber
\end{eqnarray}
According to the dominated convergence theorem, $\nabla_{\theta} D_{\theta}(v)$ is continuous in $\theta$ in a neighborhood of $\theta_0$ for all $v$.

Next, we show the second statement.
Because $D_{\theta}(v)$ is continuously differentiable in $\theta$, for $\theta$ in a neighborhood of $\theta_0$, there exists $\bar{\theta}_v$ between $\theta$ and $\theta_0$ such that
\begin{eqnarray}
 |D_{\theta}(v) - D_{\theta_0}(v)- \Gamma_0(v)'(\theta - \theta_0)| &=& | \{ \Gamma_{\bar{\theta}_v}(v) - \Gamma_0(v) \} (\theta - \theta_0) | \nonumber \\
&\leq & \|\theta - \theta_0\| \times \sup_{v \in \mathcal{V}} \| \Gamma_{\bar{\theta}_v}(v) - \Gamma_0(v) \|  \nonumber
\end{eqnarray}
It follows from Assumption N3(ii) that $\sup_{v \in \mathcal{V}} \| \Gamma_{\bar{\theta}_v}(v) - \Gamma_0(v) \| \rightarrow 0$ as $\|\theta - \theta_0\|\rightarrow 0$.
Hence, we have $|D_{\theta}(v) - D_{\theta_0}(v) - \Gamma_0(v)'(\theta - \theta_0)| = o(\|\theta - \theta_0\|)$ uniformly over $v$.
\end{proof}
\vspace{0.1in}

\begin{Lemma}
Under Assumptions C3 and C5, for every sequence $\{\delta_n\}$ of positive numbers that converges to zero,
\begin{equation}
\sup_{\|\theta - \theta_0\|<\delta_n} \|\hat{D}_{n,\theta} - D_{\theta} -\hat{D}_{n,\theta_0}\|_{\mu} = o_p(n^{-1/2}). \label{stochastic equicontinuity D}
\end{equation}
\end{Lemma}

\begin{proof}
Note that
$$
\sup_{\|\theta - \theta_0\|<\delta_n} \|\hat{D}_{n,\theta} - D_{\theta} -\hat{D}_{n,\theta_0}\|_{\mu} \leq \sup_{v \in \mathcal{V}, \|\theta - \theta_0\|<\delta_n} \left| \hat{D}_{n,\theta}(v) - D_{\theta}(v) -\hat{D}_{n,\theta_0}(v) \right|.
$$
By Lemma 2, the right-hand side is $o_p(n^{-1/2})$.
Hence, (\ref{stochastic equicontinuity D}) holds.
\end{proof}
\vspace{0.1in}

\begin{Lemma}
Under the assumptions for Theorem 3, $\hat{\theta}^* \rightarrow_{a.s} \theta_0$ for almost all samples $W_1, \cdots, W_n$.
\end{Lemma}
\begin{proof}
By the triangle inequality, for almost all samples $W_1, \cdots , W_n$, we have
\begin{eqnarray}
\sup_{\theta}\|\hat{D}_{n,\theta}^* - D_{\theta} \|_{\mu} &\leq & \sup_{\theta}\|\hat{D}_{n,\theta}^* - \hat{D}_{n,\theta}\|_{\mu} + \sup_{\theta}\|\hat{D}_{n,\theta} - D_{\theta}\|_{\mu} \nonumber \\
&\leq & \sup_{\theta, v} | \hat{D}_{n,\theta}^*(v) - \hat{D}_{n,\theta}(v)| + \sup_{\theta , v} |\hat{D}_{n,\theta}(v) - D_{\theta}(v)| \rightarrow_{a.s.} 0, \nonumber
\end{eqnarray}
since $\{A_{\theta}^v: \theta \in \Theta, v \in \mathcal{V} \}$ is a Donsker class.
The reminder of the proof is same as for Theorem 3.
\end{proof}
\vspace{0.1in}

\begin{Lemma}
Suppose that the assumptions of Theorem 4 hold.
For every sequence $\{\delta_n\}$ of positive numbers that converges to zero, 
\begin{equation}
\sup_{\|\theta - \theta_0\|<\delta_n} \| \sqrt{n} (\hat{D}^*_{n,\theta}-\hat{D}_{n,\theta}) - \sqrt{n} (\hat{D}^*_{n,\theta_0}-\hat{D}_{n,\theta_0}) \|_{\mu}=o_B(1). \label{bootstrap stochastic equicontinuity D}
\end{equation}
\end{Lemma}
\begin{proof}
The left-hand side of (\ref{bootstrap stochastic equicontinuity D}) is dominated above by
$$
\sup_{v, \|\theta - \theta_0\|<\delta_n} | \sqrt{n} (\hat{D}^*_{n,\theta}(v)-\hat{D}_{n,\theta}(v)) - \sqrt{n} (\hat{D}^*_{n,\theta_0}(v)-\hat{D}_{n,\theta_0}(v)) |.
$$
The bootstrap equicontinuity due to \cite{gine1990bootstrapping} implies that this random variable is $o_B(1)$.
Hence, we can obtain (\ref{bootstrap stochastic equicontinuity D}).
\end{proof}

\afterpage{\clearpage}
\newpage
\section*{Appendix 4: Figures and Tables}

\begin{figure}[h]
\centering
\includegraphics[width=16cm]{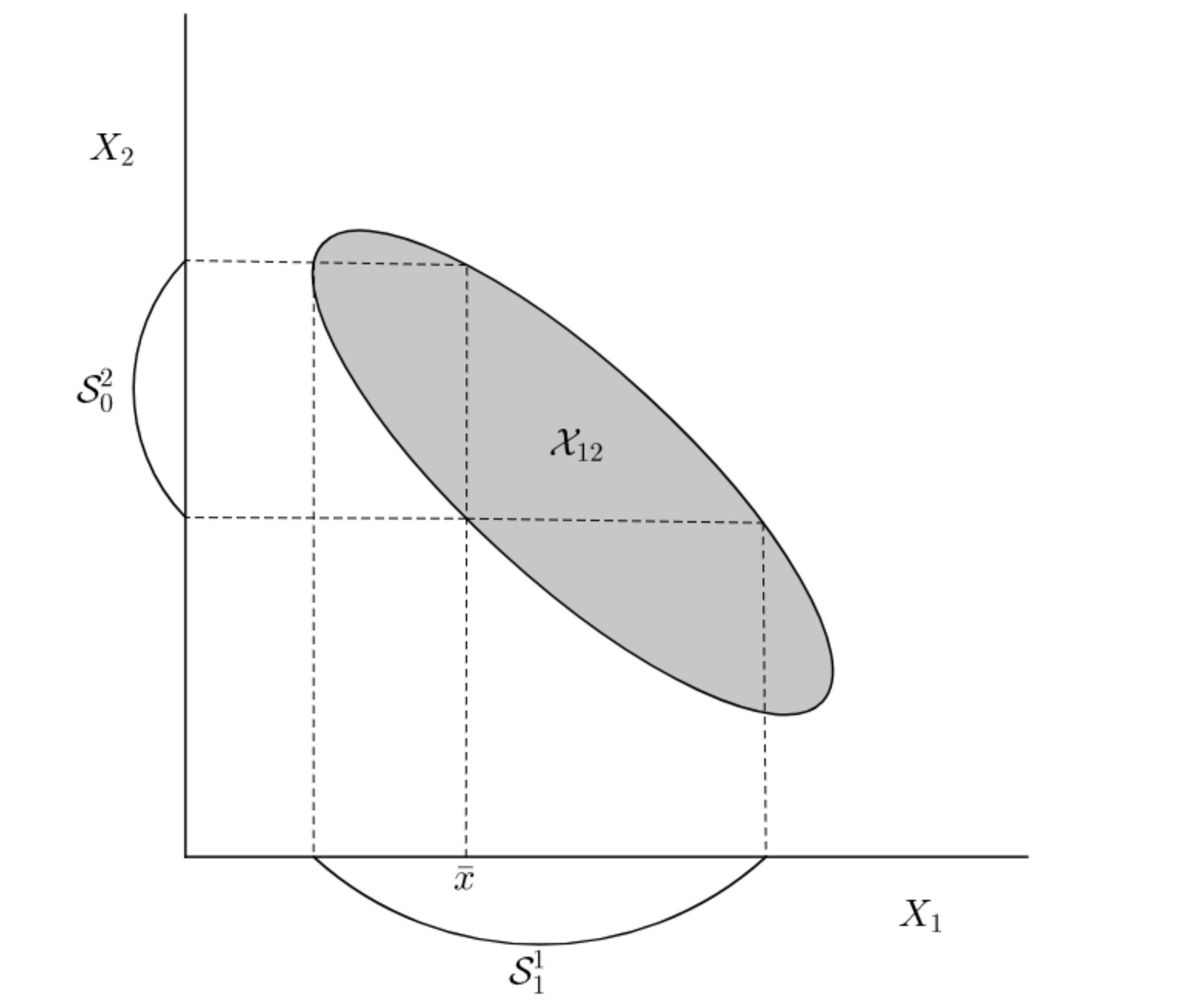}
\caption{Description of $\mathcal{S}_m^t$.}
\end{figure}

\begin{figure}
\centering
\includegraphics[width=16cm]{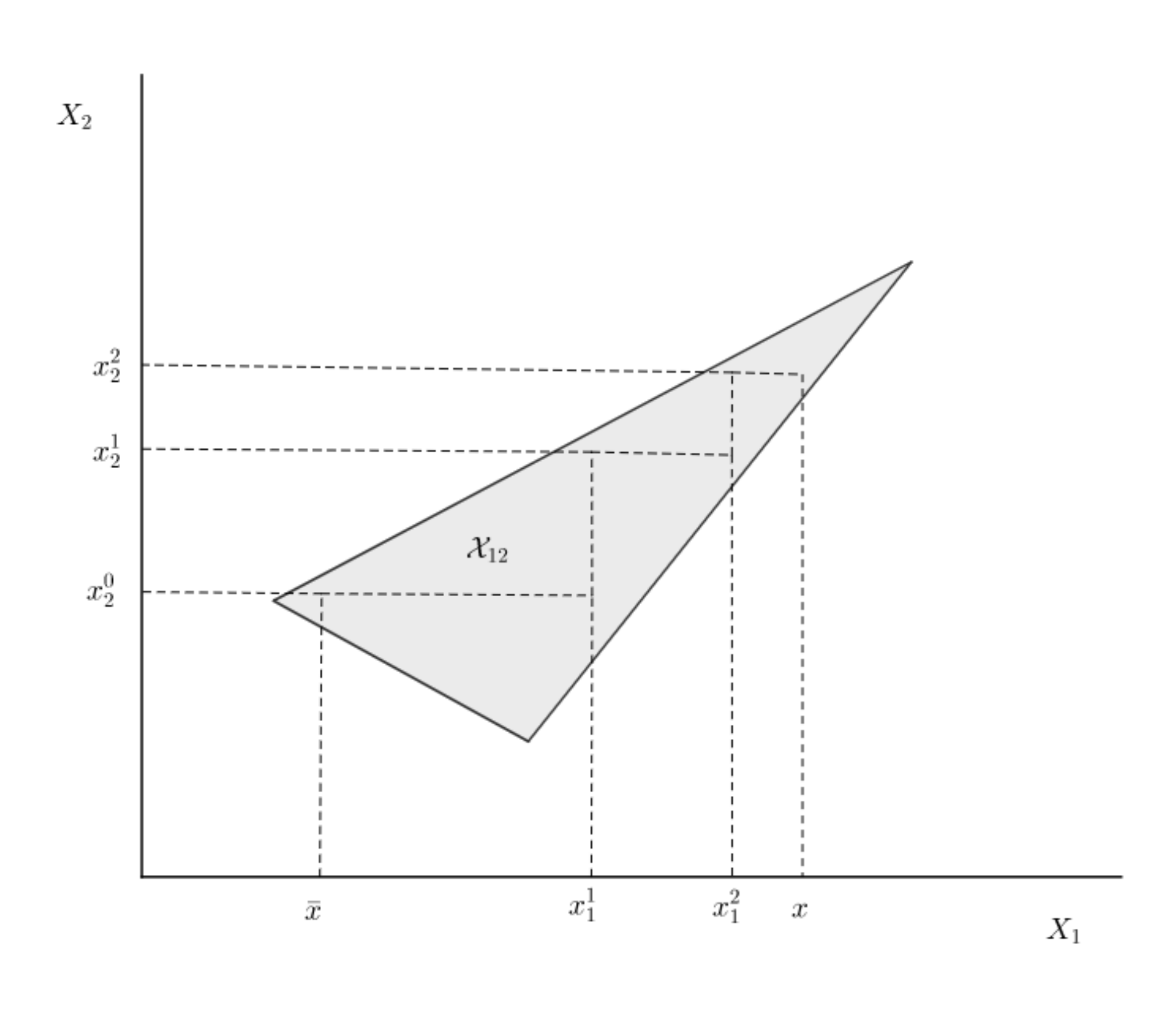}
\caption{Connected support.}
\end{figure}

\begin{table}[htb]
\begin{center}
\caption{Results of Simulation 1(i)}\vspace{0.1in}
  \begin{tabular}{c c r r r} \hline
      & & $N=400$ & $N=800$ & $N=1600$  \\ \hline
      & bias & 0.0262 & 0.0195 & 0.0125 \\ 
$\theta_1$ & std & 0.1991 & 0.1709 & 0.1121 \\ 
      & mse & 0.0403 & 0.0296 & 0.0127 \\ \hline
      & bias & 0.0657 & 0.0733 & 0.0363 \\ 
$\theta_2$ & std & 0.1751 & 0.1375 & 0.1045 \\ 
      & mse & 0.0350 & 0.0243 & 0.0122 \\ \hline
      & bias & 0.0472 & 0.0605 & 0.0320 \\ 
$\theta_3$ & std & 0.1158 & 0.1069 & 0.0889 \\ 
      & mse & 0.0231 & 0.0171 & 0.0089 \\ \hline
  \end{tabular}
\end{center}
\end{table}

\begin{table}[htb]
\begin{center}
\caption{Results of Simulation 1(ii)}\vspace{0.1in}
  \begin{tabular}{c c r r} \hline
      & & $N=500$ & $N=1000$   \\ \hline
      & bias & 0.0101 & 0.0087  \\ 
Our method & std & 0.0978 & 0.0745  \\ 
      & mse & 0.0097 & 0.0056 \\ \hline
      & bias & 0.0042 & 0.0017  \\ 
Hoderlein and White & std & 0.1175 & 0.0855  \\ 
      & mse & 0.0138 & 0.0073  \\ \hline
  \end{tabular}
\end{center}
\end{table}

\begin{table}[htb]
\begin{center}
\caption{Results of Simulation 2} \vspace{0.1in}
  \begin{tabular}{c c r r r} \hline
      & & $N=400$ & $N=800$ & $N=1600$  \\ \hline 
      & bias & -0.0025 & -0.0031 & 0.0005 \\ 
$\theta_1$ & std & 0.0875 & 0.0604 & 0.0418 \\ 
      & mse & 0.0077 & 0.0037 & 0.0018 \\ \hline
      & bias & 0.0020 & 0.0022 & 0.0001 \\ 
$\theta_2$ & std & 0.0522 & 0.0351 & 0.0244 \\ 
      & mse & 0.0027 & 0.0012 & 0.0006 \\ \hline
      & bias & -0.0082 & 0.0004 & -0.0044 \\ 
$\theta_3$ & std & 0.1837 & 0.1256 & 0.0886 \\ 
      & mse & 0.0338 & 0.0158 & 0.0079 \\ \hline
      & bias & 0.0043 & 0.0001 & 0.0023 \\ 
$\theta_4$ & std & 0.0839 & 0.0562 & 0.0398 \\ 
      & mse & 0.0071 & 0.0032 & 0.0016 \\ \hline
      & bias & -0.0025 & -0.0014 & -0.0011 \\ 
ATE & std & 0.0972 & 0.0673 & 0.0457 \\ 
      & mse & 0.0095 & 0.0045 & 0.0021 \\ \hline
      & bias & -0.0333 & -0.0303 & -0.0343 \\ 
QTE25 & std & 0.1148 & 0.0814 & 0.0565 \\ 
      & mse & 0.0143 & 0.0075 & 0.0044 \\ \hline
      & bias & -0.0017 & -0.0016 & -0.0009 \\ 
QTE50 & std & 0.0994 & 0.0685 & 0.0474 \\ 
      & mse & 0.0099 & 0.0047 & 0.0022 \\ \hline
      & bias & 0.0306 & 0.0292 & 0.0328 \\ 
QTE75 & std & 0.1324 & 0.0896 & 0.0618 \\ 
      & mse & 0.0185 & 0.0089 & 0.0049 \\ \hline
  \end{tabular}
\end{center}
\end{table}

\begin{table}[htb]
  \begin{center}
  \caption{Lower and upper bounds for Example 4}
    \begin{tabular}{c c c c c c c } \hline
    $x$ & $g_1^L(x,u)$ & $g_1^U(x,u)$ & $g_2^L(x,u)$ & $g_2^U(x,u)$ & $g_1(x,u)$ & $g_2(x,u)$ \\ \hline 
    $-2$ & $\mathbf{1}\{u > 0.99\}$ & $\mathbf{1}\{u > 0.47\}$ & $\mathbf{1}\{u > 0.99\}$ & $\mathbf{1}\{u > 0.41\}$ & $\mathbf{1}\{u>0.73\}$ & $\mathbf{1}\{u>0.71\}$ \\ 
    $-1$ & $\mathbf{1}\{u > 0.98\}$ & $\mathbf{1}\{u > 0.35\}$ & $\mathbf{1}\{u > 0.98\}$ & $\mathbf{1}\{u > 0.30\}$ & $\mathbf{1}\{u>0.62\}$ & $\mathbf{1}\{u>0.57\}$ \\ 
    $0$ & $\mathbf{1}\{u > 0.85\}$ & $\mathbf{1}\{u > 0.13\}$ & $\mathbf{1}\{u > 0.82\}$ & $\mathbf{1}\{u > 0.10\}$ & $\mathbf{1}\{u>0.50\}$ & $\mathbf{1}\{u>0.43\}$ \\ 
    $1$ & $\mathbf{1}\{u > 0.64\}$ & $\mathbf{1}\{u > 0.02\}$ & $\mathbf{1}\{u > 0.59\}$ & $\mathbf{1}\{u>0.01\}$ & $\mathbf{1}\{u>0.38\}$ & $\mathbf{1}\{u>0.29\}$ \\ 
    $2$ & $\mathbf{1}\{u > 0.53\}$ & $\mathbf{1}\{u>0.01\}$ & $\mathbf{1}\{u > 0.47\}$ & $\mathbf{1}\{u>0.01\}$ & $\mathbf{1}\{u>0.27\}$ & $\mathbf{1}\{u>0.18\}$ \\ \hline
    \end{tabular}
  \end{center}
\end{table}

\clearpage

\bibliographystyle{ecta}
\bibliography{nonseparable_model}

\end{document}